\documentclass[journal]{IEEEtran}
\usepackage{amsmath}
\usepackage{amsthm}
\usepackage{amsfonts}
\usepackage{amssymb}
\usepackage{multicols}
\usepackage{multirow}
\usepackage{authblk}
\usepackage{tikz}
\usepackage{graphicx}
\usepackage{epstopdf}
\usepackage{mathtools}

\theoremstyle{remark}
\newtheorem{thm}{Theorem}
\newtheorem{conj}{Conjecture}
\newtheorem{lem}{Lemma}

\newtheorem{cor}{Corollary}
\newtheorem{defn}{Definition}
\newtheorem{exmp}{Example}

\newtheorem{rem}{Remark}

\newtheorem{obs}{Observation}

\title{On the Optimal Broadcast Rate of the Two-Sender Unicast Index Coding Problem with Fully-Participated Interactions}
\author{Chinmayananda Arunachala, Vaneet Aggarwal, and B. Sundar Rajan. \thanks{C. Arunachala and B. S. Rajan are with the Dept.
		of Electrical Communication Engg., Indian Institute of Science, Bengaluru
		560012, KA, India, email: \{chinmayanand,bsrajan\}@iisc.ac.in. 	V. Aggarwal is with the School of Industrial
		Engineering at Purdue University, West Lafayette, IN, USA 47907. He is also with the 
		Dept. of Electrical Communication Engg., Indian Institute of Science, Bengaluru
		560012, KA, India, email:
		vaneet@purdue.edu.}}

\begin{document}
\maketitle

\begin{abstract}
	The problem of two-sender unicast index coding consists of two senders and a set of receivers. Each receiver demands a unique message and possesses some of the messages demanded by other receivers as its side-information. Every demanded message is present with at least one of the senders. Senders avail the knowledge of the side-information at the receivers to reduce the number of broadcast transmissions. Solution to this problem consists of finding the optimal number of coded transmissions from the two senders. One important class of the two-sender problem consists of the messages at the senders and the side-information at the receivers satisfying  \emph{fully-participated interactions}. This paper provides the optimal broadcast rates, for all the unsolved cases of the two-sender problem with fully-participated interactions when the associated \emph{interaction digraphs} contain cycles. The optimal broadcast rates are provided in terms of those of the three independent single-sender problems associated with the two-sender problem. This paper also provides an achievable  broadcast rate with $t$-bit messages for any finite $t$ and any two-sender problem with fully-participated interactions belonging to $(i)$ any one of the six instances (classes) of the two-sender problem when the associated interaction digraph does not contain any cycle, and $(ii)$ one of the classes of the two-sender problem when the associated interaction digraph contains cycles. The achievable broadcast rates are obtained by exploiting the symmetries of the confusion graph to color the same according to the two-sender graph coloring.    
\end{abstract}

\section{Introduction}
Index coding problem (ICP) with a single sender and a set of receivers was introduced in \cite{BK}. Each receiver has some messages as its side-information and demands a unique message. The sender avails the cognizance of the side-information present at all the receivers to reduce the number of broadcast transmissions, such that all the receivers can decode their demands using the broadcast transmissions and their side-information. Many practical scenarios demand for distributed transmissions, where the messages are distributed among multiple senders. Content is delivered in cellular networks using large storage capacity nodes called caching helpers \cite{KAC}, where the messages are distributed among the helpers to reduce the total average delay of all the users. Data is stored over multiple storage nodes to account for any failure in one or more storage nodes in distributed storage networks \cite{luo2016coded,xiang2016joint}. Hence, multi-sender ICP is of practical significance. 

Multi-sender ICPs were first studied by Ong et al. in \cite{SUOH}. They studied a special class of multi-sender ICPs, where each receiver knows a unique message and demands a subset of other messages. An iterative algorithm was proposed which provided different lower bounds on the optimal codelength based on the \emph{strongly-connected component} of the \emph{information-flow graph} selected in each iteration. Tightness of the lower bounds for the optimal codelength was not quantified and no performance guarantees were given for the algorithm \cite{SUOH}. Two-sender unicast ICPs were studied by Thapa et al. \cite{COJ}. They extended some graph theory based single-sender index coding schemes to the two-sender unicast ICP (TUICP). No class of non-trivial  two-sender problems were identified for which the proposed schemes were optimal. Tightness of the gap between the optimal codelengths and the codelengths given by the proposed schemes were not quantified for any special class of the TUICP. Many variations of multi-sender ICPs were studied and bounds for the capacity region were given \cite{sadeghi2016distributed,YPFK,MOJ2,MOJ}. These works assume that there are independent channels with fixed finite capacities from every sender to every receiver. This is in contrast with the previous works where a single shared channel was assumed with the  transmissions from   multiple senders being orthogonal in time. Variations of random coding were used to provide the bounds and the bounds were tightened with further improvements in the encoding schemes. 

Thapa et al. \cite{CTLO} studied the TUICP using a new variation of graph coloring called the two-sender graph coloring to account for the non-availability of some messages at each sender. The \emph{confusion graph} was colored according to the two-sender graph coloring to obtain the optimal broadcast rate with $t$-bit messages for any finite $t$. The TUICP described by the \emph{side-information digraph} has been analyzed using three independent single-sender sub-problems  described by the vertex-induced sub-digraphs of the side-information digraph and the \emph{interactions} between these sub-problems. The TUICP was divided into 64 classes based on these \emph{interactions}. The type of interaction was captured using the \emph{interaction digraph} obtained from the associated side-information digraph. There are 64 possibile interaction digraphs broadly classified into two cases: Case I and Case II. Case I consists of acyclic interaction digraphs and the optimal broadcast rate with $t$-bit messages for any finite $t$ was obtained using the two-sender graph coloring of the confusion graph for half of the sub-cases with any type of interactions (\emph{fully-participated} or \emph{partially-participated}). For the remaining sub-cases, a conjecture was stated. Case II was further classified into five sub-cases. For Case II-A, the optimal broadcast rate with $t$-bit messages for any finite $t$ was obtained for any two-sender problem with any type of interactions. For other sub-cases, only fully-participated interactions were considered. For Cases II-C, II-D, and II-E, the optimal broadcast rate with $t$-bit messages for any finite $t$ was obtained for some sub-cases based on the relation between the corresponding optimal broadcast rates of the sub-problems. For other sub-cases, upper bounds were provided. The optimal broadcast rates with $t$-bit messages for any finite $t$, the corresponding code constructions and the upper bounds (for the optimal broadcast rates) were provided in terms of those of the three related single-sender unicast index coding sub-problems.

Optimal scalar linear codes were obtained for some special sub-cases of Case II-C and Case II-D with partially-participated interactions using the notion of joint extensions of two single-sender index coding problems  \cite{CBSR}. Optimal linear broadcast rates with $t$-bit messages for any finite $t$ and the corresponding code constructions were provided for all the cases of the two-sender problem with fully-participated interactions in \cite{CVBSR1}.

\begin{table}[h]
	\centering
	\begin{tabular}{|c|c|c|c|c|}
		\hline
		CASE & 
		$\beta_t(\mathcal{D}^k,\mathcal{P})$ \\  \hline
		I,$\mathcal{H}_{16}$ & $\beta_t(\mathcal{D}_{2}^{16,\mathcal{P}})+\beta_t(\mathcal{D}_{1*3}^{16,\mathcal{P}})$ \\ \hline
		I,$\mathcal{H}_{18}$ & $\beta_t(\mathcal{D}_{2}^{18,\mathcal{P}})+\beta_t(\mathcal{D}_{1 \circ 3}^{18,\mathcal{P}})$ \\ \hline
		I,$\mathcal{H}_{20}$ & $\beta_t(\mathcal{D}_{1}^{20,\mathcal{P}})+\beta_t(\mathcal{D}_{2*3}^{20,\mathcal{P}})$ \\ \hline
		I,$\mathcal{H}_{21}$ & $\beta_t(\mathcal{D}_{1}^{21,\mathcal{P}})+\beta_t(\mathcal{D}_{2 \circ 3}^{21,\mathcal{P}})$ \\ \hline
		I,$\mathcal{H}_{23}$ & $\beta_t(\mathcal{D}_{2}^{23,\mathcal{P}})+\beta_t(\mathcal{D}_{3 \circ 1}^{23,\mathcal{P}})$ \\ \hline
		I,$\mathcal{H}_{25}$ &  $\beta_t(\mathcal{D}_{1}^{25,\mathcal{P}})+\beta_t(\mathcal{D}_{3 \circ 2}^{25,\mathcal{P}})$ \\ \hline
		\multirow{2}{*}{II-E,$\mathcal{H}_{k}$}  & 
		\multirow{2}{*}[.5ex]{$max\{\beta_t(\mathcal{D}_{1}^{k,\mathcal{P}})+\beta_t(\mathcal{D}_{2}^{k,\mathcal{P}}),\beta_t(\mathcal{D}_{1}^{k,\mathcal{P}})$} \\[5pt]
		\multirow{2}{*}{} & \multirow{2}{*}[1ex]{$+\beta_t(\mathcal{D}_{3}^{k,\mathcal{P}}),\beta_t(\mathcal{D}_{2}^{k,\mathcal{P}})+\beta_t(\mathcal{D}_{3}^{k,\mathcal{P}})\}$} \\[5pt] \hline
	\end{tabular}	
	\vspace{5pt}
	\caption[LoF entry]{Summary of the achievability results for any tuicp with fully-participated
		interactions between the sub-digraphs $\mathcal{D}_{1}^{k,\mathcal{P}},\mathcal{D}_{2}^{k,\mathcal{P}}$, and $\mathcal{D}_{3}^{k,\mathcal{P}}$ in the side-information digraph $\mathcal{D}^{k}$. 
		
		The definitions of the digraphs $\mathcal{D}_{i \circ  j}^{k,\mathcal{P}}$ and $\mathcal{D}_{i * j}^{k,\mathcal{P}}$ are given in the definitions \ref{deflexi1} and \ref{defdisj1} respectively. The listing of all the interaction digraphs $\mathcal{H}_k$ and the related cases is given in Figure \ref{interenum}. }
	\label{table1}
	\vspace{-15pt}
\end{table}

In this paper, we first provide an achievable broadcast rate with $t$-bit messages for any finite $t$ and any two-sender problem belonging to any one of the six sub-cases of Case I of the TUICP with fully-participated interactions, for which no achievable broadcast rates were given earlier. This is obtained by providing a valid  two-sender graph coloring of the confusion graph associated with the two-sender problem. This exploits some unexplored symmetries of the confusion graph. We propose new ways of grouping its vertices to color it according to the two-sender graph coloring. The proofs indicate a possibility for the Conjecture 1, given in \cite{CTLO} to be false. However, we don't disprove the conjecture.  We also provide an achievable broadcast rate with $t$-bit messages for any finite $t$ and any two-sender problem belonging to Case II-E of the TUICP with fully-participated interactions. This is obtained using a code-construction for the two-sender problem using codes of the associated single-sender sub-problems. This serves as a tighter upper bound on the optimal broadcast rate with $t$-bit messages for any finite $t$ compared to those given in \cite{CTLO}. The results on the achievable broadcast rates presented in this paper are given in Table \ref{table1}. All the notations and definitions required to understand the results are given in Sections II and III.

We then provide the optimal broadcast rates for all the two-sender problems belonging to Case II with fully-participated interactions, for which only upper bounds were provided earlier. The optimal broadcast rates are given in terms of those of the three related single-sender unicast index coding sub-problems. Thus, the complexity of finding the optimal broadcast rate for the TUICP with fully-participated interactions is reduced to that of finding the optimal broadcast rate for the single-sender unicast ICP, which is an NP-Hard problem in general. We use the existing results available for the single-sender unicast index coding problem to prove tight lower bounds on the optimal broadcast rates based on the two-sender graph coloring of the confusion graphs. The matching upper bounds were provided by Thapa et al. \cite{CTLO} for problems belonging to Cases II-C and II-D. For Case II-E, we utilize the achievable broadcast rate with $t$-bit messages for any finite $t$ given in this paper, to obtain an upper bound on the optimal broadcast rate which matches the lower bound obtained by using the results of Cases II-C and II-D. All the results on the optimal broadcast rates of any TUICP with fully-participated interactions are given in Table \ref{table2}. All the notations and definitions required to understand the results are given in Sections II and III. The results marked with a $``!"$ are the ones which also hold for any  partially-participated interactions. The results marked by $``*"$ are the ones which are provided in this paper and was partially resolved in \cite{CTLO}. 

\begin{table}[h]
	\centering
	\begin{tabular}{|c|c|c|c|c|}
		\hline
		CASE & 
		$\beta(\mathcal{D}^k,\mathcal{P})$ \\  \hline
		I    & $\beta(\mathcal{D}_{1}^{k,\mathcal{P}})+\beta(\mathcal{D}_{2}^{k,\mathcal{P}})+\beta(\mathcal{D}_{3}^{k,\mathcal{P}})!$ \\ \hline
		II-A & $\beta(\mathcal{D}_{1}^{k,\mathcal{P}})+\beta(\mathcal{D}_{2}^{k,\mathcal{P}})+\beta(\mathcal{D}_{3}^{k,\mathcal{P}})!$ \\ \hline
		II-B & $max\{\beta(\mathcal{D}_{3}^{k,\mathcal{P}}),\beta(\mathcal{D}_{1}^{k,\mathcal{P}})+\beta(\mathcal{D}_{2}^{k,\mathcal{P}})\}$ \\ \hline
		II-C & $\beta(\mathcal{D}_{2}^{k,\mathcal{P}})+max\{\beta(\mathcal{D}_{1}^{k,\mathcal{P}}),\beta(\mathcal{D}_{3}^{k,\mathcal{P}})\}*$ \\ \hline
		II-D & $\beta(\mathcal{D}_{1}^{k,\mathcal{P}})+max\{\beta(\mathcal{D}_{2}^{k,\mathcal{P}}),\beta(\mathcal{D}_{3}^{k,\mathcal{P}})\}*$ \\ \hline
		\multirow{2}{*}{II-E} & 
		\multirow{2}{*}[.5ex]{$max\{\beta(\mathcal{D}_{1}^{k,\mathcal{P}})+\beta(\mathcal{D}_{2}^{k,\mathcal{P}}),\beta(\mathcal{D}_{1}^{k,\mathcal{P}})$} \\[5pt]
		\multirow{2}{*}{} & \multirow{2}{*}[1ex]{$+\beta(\mathcal{D}_{3}^{k,\mathcal{P}}),\beta(\mathcal{D}_{2}^{k,\mathcal{P}})+\beta(\mathcal{D}_{3}^{k,\mathcal{P}})\}*$} \\[5pt] \hline
	\end{tabular}	
	\vspace{5pt}
	\caption[LoF entry]{Optimal broadcast rates for any tuicp with fully-participated interactions between the sub-digraphs  $\mathcal{D}_{1}^{k,\mathcal{P}},\mathcal{D}_{2}^{k,\mathcal{P}}$, and $\mathcal{D}_{3}^{k,\mathcal{P}}$ in the side-information digraph $\mathcal{D}^{k}$. The listing of all the cases is given in Figure \ref{interenum}. }
	\label{table2}
	\vspace{-15pt}
\end{table}

The key results of this paper are summarized as follows. 
\begin{itemize}
	\item Achievable broadcast rates with $t$-bit messages for any finite $t$ is given for six sub-cases of the TUICP belonging to Case I with \emph{fully-participated interactions}, for which no non-trivial achievable schemes were known prior to this work. 
	\item Achievable broadcast rates with $t$-bit messages for any finite $t$ is given for Case II-E of the TUICP with \emph{fully-participated interactions}, which serve as tighter upper bounds for the optimal  broadcast rates with $t$-bit messages for any finite $t$, compared to those known prior to this work.
	\item Optimal broadcast rates are established for all the sub-cases of Case II with \emph{fully-participated interactions}, for which only upper bounds were known earlier.
\end{itemize}
\par The remainder of the paper is organized as follows. Section II introduces the problem setup and establishes the required definitions and notations. Section III recapitulates the notion of the confusion graph and the two-sender graph coloring of the same. Section IV provides achievable broadcast rates with finite-length messages, for some sub-cases of the two-sender problem with fully-participated interactions belonging to Case I and all sub-cases of Case II-E. Section V provides optimal broadcast rates for all the TUICPs with fully-participated interactions belonging to Cases II-C, II-D, and II-E. Section VI concludes the paper.

\section{Problem Formulation and Definitions}

In this section, we formulate the two-sender unicast index coding problem, and establish the  notations and definitions used in this paper. 
\par  The set $\{1,2,\cdots, n\}$ is denoted as $[n]$. In a two-sender unicast index coding problem (TUICP), there are $m$ independent messages given by the set $\mathcal{M} =\{{\bf{x}}_1,{\bf{x}}_2,\cdots,{\bf{x}}_{m}\}$, where ${\bf{x}}_i \in \mathbb{F}_2^{t \times 1}$, $\forall i \in [m]$, and $t \geq 1$. There are $m$ receivers. The $i$th receiver demands ${\bf{x}}_i$ and has $\mathcal{K}_i \subseteq \mathcal{M} \setminus \{{\bf{x}}_i\}$ as its side-information. The $s$th sender is denoted by $\mathcal{S}_{s}$, $s \in \{1,2\}$. $\mathcal{S}_{s}$ possesses the message set  $\mathcal{M}_{s}$, such that $\mathcal{M}_{s} \subset \mathcal{M}$, and $\mathcal{M}_{1} \cup \mathcal{M}_{2}=\mathcal{M}$. Each sender knows the identity of the messages present with the other sender. The senders transmit through a noiseless broadcast channel and the transmissions from different senders are orthogonal in time.  Single-sender unicast ICP is a special case of TUICP, where $\mathcal{M}_{1}=\mathcal{M}$ and  $\mathcal{M}_{2}=\Phi$. 

Given an instance of the TUICP, each codeword of a two-sender index code consists of two sub-codewords broadcasted by the two senders respectively. An encoding function for the sender $\mathcal{S}_{s}$ is given by $\mathbb{E}_{s}:\mathbb{F}_{2}^{|\mathcal{M}_{s}|t \times 1} \rightarrow   \mathbb{F}_{2}^{p_{s} \times 1}$, such that $\mathcal{C}_s=\mathbb{E}_{s}(\mathcal{M}_s)$, where $p_s$ is the length  of the sub-codeword $\mathcal{C}_s$ transmitted by $\mathcal{S}_s$, $s \in \{1,2\}$. The sub-codewords from the two senders are sent one after the other. The $i$th receiver has a decoding function given by $\mathbb{D}_{i}:\mathbb{F}_{2}^{(p_{1}+p_{2}+|\mathcal{K}_{i}|t) \times 1} \rightarrow   \mathbb{F}_{2}^{t \times 1}$, such that ${\bf{x}}_i = \mathbb{D}_{i}(\mathcal{C}_1,\mathcal{C}_2,\mathcal{K}_i)$, $\forall i \in [m]$, i.e., it can decode ${\bf{x}}_i$ using its side-information and the received codeword consisting of $\mathcal{C}_1$ and $\mathcal{C}_2$. For single-sender unicast ICP, $\mathcal{M}_{2}=\Phi$.  Hence, $p_2=0$. In this case, we assume that only $\mathbb{E}_1$ exists.

\par We state the definitions of broadcast rate of an index code, the optimal broadcast rate of a two-sender problem with $t$-bit messages for any  finite $t$, and the optimal broadcast rate of the same, as given in \cite{CTLO}. Note that the definitions take into account both linear and non-linear encoding schemes.
\begin{defn}[Broadcast rate, \cite{CTLO}] 
	The broadcast rate of an index code (for a single-sender problem or a two-sender problem) described by $(\{\mathbb{E}_{j}\},\{\mathbb{D}_{i}\})$ is the total number of transmitted bits  per received message bits ($t$-bit messages for some finite $t$), given by $p_{t} \triangleq \frac{(p_{1}+p_{2})}{t}$. 
	
	The optimal (minimum) length of any index code for a given ICP and $t$-bit messages is called the optimal codelength.
	
\end{defn}
\begin{defn}[Optimal broadcast rate with $t$-bit messages for any finite $t$, \cite{CTLO}]
	The optimal broadcast rate for a given ICP with $t$-bit messages and any finite $t$ is given by $\beta_{t} \triangleq \underset{{\{\mathbb{E}_j\}}}{min}$  $p_{t}$. 
\end{defn}
Note that for a single-sender unicast ICP, only $\mathbb{E}_1$ is considered in the expression for $\beta_{t}$.
\begin{defn}[Optimal broadcast rate, \cite{CTLO}]
	The optimal broadcast rate of a given ICP is given by   $\beta \triangleq \underset{t}{inf} \beta_{t} = \underset{t \rightarrow \infty}{lim} \beta_{t}$.
\end{defn}

\par We state some definitions from graph theory \cite{DBW}, that will be used in this paper. 

A directed graph (also called digraph) given by $\mathcal{D}=(\mathcal{V}(\mathcal{D}),\mathcal{E}(\mathcal{D}))$, consists of a set of vertices $\mathcal{V}(\mathcal{D})$, and a set of edges $\mathcal{E}(\mathcal{D})$ which is a set of ordered pairs of vertices. A sub-digraph  $\mathcal{G}$ of a digraph $\mathcal{D}$ is a digraph, whose vertex set satisfies $\mathcal{V}(\mathcal{G}) \subseteq \mathcal{V}(\mathcal{D})$, and edge set satisfies $\mathcal{E}(\mathcal{G}) \subseteq \mathcal{E}(\mathcal{D})$. The sub-digraph of $\mathcal{D}$ induced by the vertex set $\mathcal{V}(\mathcal{G})$ is the  digraph whose vertex set is $\mathcal{V}(\mathcal{G})$, and the edge set is given by  $\mathcal{E}(\mathcal{G})=\{(u, v): u,v \in \mathcal{V}(\mathcal{G}), (u, v) \in \mathcal{E}(\mathcal{D})\}$. A directed path in a digraph $\mathcal{D}$ is a sequence of distinct vertices $\{v_1,\cdots,v_r\}$, such that $(v_i,v_{i+1}) \in \mathcal{E}(\mathcal{D})$, $\forall i \in [r-1]$. A cycle in a digraph $\mathcal{D}$ is a sequence of distinct vertices $(v_{1},\cdots,v_{c})$, such that $(v_{i},v_{i+1}) \in \mathcal{E}(\mathcal{D})$,  $\forall i \in [c-1]$, and $(v_{c},v_{1}) \in \mathcal{E}(\mathcal{D})$. 

For an undirected graph, the edge set consists of a set of unordered pairs of vertices. Two vertices are said to be adjacent if there exists an edge between the two vertices. A proper graph coloring of an undirected graph $\mathcal{D}$ is an onto function $J: \mathcal{V}(\mathcal{D}) \rightarrow \mathcal{J}$, where $\mathcal{J}$ is a set of colors such that, if $(u,v) \in \mathcal{E}(\mathcal{D})$, then $J(u) \neq J(v)$. The minimum number of colors required for any proper coloring of an undirected graph $\mathcal{D}$ is its chromatic number and is denoted by  $\chi(\mathcal{D})$. Two undirected graphs $\mathcal{G}$ and $\mathcal{H}$ are said to be isomorphic if there exists a bijection between the vertex sets $\mathcal{V}(\mathcal{G})$ and $\mathcal{V}(\mathcal{H})$ given by $f: \mathcal{V}(\mathcal{G}) \rightarrow \mathcal{V}(\mathcal{H})$, such that $(u,v) \in \mathcal{E}(\mathcal{G})$ iff $(f(u),f(v)) \in \mathcal{E}(\mathcal{H})$. A subset of vertices of an undirected graph $\mathcal{V}(\mathcal{G})$ is said to be independent, if there is no edge between any two vertices of the subset. A clique $\mathcal{C}$ of an undirected graph $\mathcal{G}$ is a vertex-induced subgraph of  $\mathcal{G}$ such that there is an edge between every pair of vertices in $\mathcal{C}$. A largest clique of a given graph is a clique with maximum number of vertices. The clique
number $\omega(\mathcal{G})$ of an undirected graph $\mathcal{G}$ is the number of vertices in a largest clique of $\mathcal{G}$. 

The following graph products of any two given undirected graphs are used in this paper. 
\begin{defn}[Lexicographic product]
	The lexicographic product $\mathcal{G}$ of two undirected graphs $\mathcal{G}_1$ and $\mathcal{G}_2$ is denoted by $\mathcal{G}_1 \circ \mathcal{G}_2$, where $\mathcal{V}(\mathcal{G}) = \mathcal{V}(\mathcal{G}_1) \times \mathcal{V} (\mathcal{G}_2)$ and
	$((u_1, u_2), (v_1,v_2)) \in \mathcal{E}(\mathcal{G})$ iff
	$(u_1, v_1) \in \mathcal{E}(\mathcal{G}_1)$ or
	$((u_1 = v_1)$ and $(u_2,v_2) \in \mathcal{E}(\mathcal{G}_2))$. 
	\label{deflexi}
\end{defn}

\begin{figure}
	\centering
	\begin{tikzpicture}
	[place/.style={circle,draw=black!100,thick}]
	\node at (-2,0) [place]  (11c) {1,1};
	\node at (0,0) [place]   (21c) {2,1};
	\node at (2,0) [place]   (31c) {3,1};
	\node at (-2,2) [place]  (12c) {1,2};
	\node[label=$\mathcal{G}_1 \circ \mathcal{G}_2$] at (0,2) [place]   (22c) {2,2};
	\node at (2,2) [place]   (32c) {3,2};
	\node at (-2,-2) [place] (13c) {1,3};
	\node at (0,-2) [place]  (23c) {2,3};
	\node at (2,-2) [place]  (33c) {3,3};
	\node[label=$\mathcal{G}_2$] at (-4,2) [place]  (2c)  {2};
	\node at (-4,0) [place]  (1c) {1};
	\node at (-4,-2) [place] (3c) {3};
	\draw [thick,-] (2c.south) -- (1c.north);
	\draw [thick,-] (1c.south) -- (3c.north);
	\draw [thick,-] (3c) to [out=135,in=-135] (2c);
	\node at (-2,4) [place]  (1'c) {1};
	\node[label=$\mathcal{G}_1$] at (0,4) [place]   (2'c)  {2};
	\node at (2,4) [place]   (3'c) {3};
	\draw [thick,-] (1'c.east) -- (2'c.west);
	\draw [thick,-] (2'c.east) -- (3'c.west);
	\draw [thick,-] (11c.south) -- (13c.north);
	\draw [thick,-] (12c.south) -- (11c.north);
	\draw [thick,-] (21c.south) -- (23c.north);
	\draw [thick,-] (22c.south) -- (21c.north);
	\draw [thick,-] (31c.south) -- (33c.north);
	\draw [thick,-] (32c.south) -- (31c.north);
	\draw [thick,-] (12c.east) -- (22c.west);
	\draw [thick,-] (11c.east) -- (21c.west);
	\draw [thick,-] (13c.east) -- (23c.west);
	\draw [thick,-] (22c.east) -- (32c.west);
	\draw [thick,-] (21c.east) -- (31c.west);
	\draw [thick,-] (23c.east) -- (33c.west);
	\draw [thick,-] (12c) to [out=-45,in=135]  (21c);
	\draw [thick,-] (11c) to [out=-45,in=135]  (23c);
	\draw [thick,-] (12c) to [out=-65,in=125]  (23c);
	\draw [thick,-] (13c) to [out=65,in=-125]  (22c);
	\draw [thick,-] (22c) to [out=-65,in=115]  (33c);
	\draw [thick,-] (23c) to [out=65,in=-115]  (32c);
	\draw [thick,-] (13c) to [out=45,in=-135]  (21c);
	\draw [thick,-] (11c) to [out=45,in=-135]  (22c);
	\draw [thick,-] (21c) to [out=45,in=-135]  (32c);
	\draw [thick,-] (23c) to [out=45,in=-135]  (31c);
	\draw [thick,-] (22c) to [out=-45,in=135]  (31c);
	\draw [thick,-] (21c) to [out=-45,in=135]  (33c);
	\draw [thick,-] (12c) to [out=-135,in=135]  (13c);
	\draw [thick,-] (22c) to [out=-115,in=115]  (23c);
	\draw [thick,-] (32c) to [out=-55,in=55]  (33c);
	\end{tikzpicture}
	\caption{Lexicographic product of  $\mathcal{G}_1$ and $\mathcal{G}_2$.}
	\label{figlexico}
\end{figure}
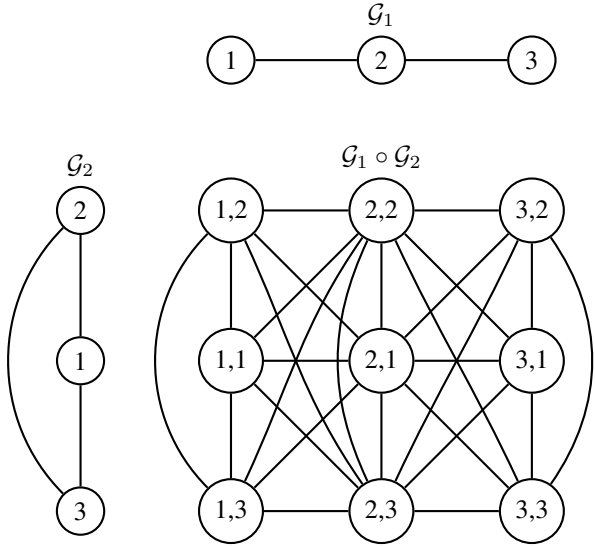

\begin{defn}[Disjunctive product]
	The disjunctive product $\mathcal{G}$  is denoted by $\mathcal{G}_1 * \mathcal{G}_2$, where  $\mathcal{V}(\mathcal{G}) = \mathcal{V}(\mathcal{G}_1) \times \mathcal{V} (\mathcal{G}_2)$ and
	$((u_1, u_2), (v_1,v_2)) \in \mathcal{E}(\mathcal{G})$ iff
	$(u_1, v_1) \in \mathcal{E}(\mathcal{G}_1)$ or $(u_2,v_2) \in \mathcal{E}(\mathcal{G}_2)$.
	\label{defdisj}
\end{defn}

\begin{figure}
	\centering
	\begin{tikzpicture}
	[place/.style={circle,draw=black!100,thick}]
	\node[label=$\mathcal{G}_2$] at (-4,2) [place]  (2c)  {2};
	\node at (-4,0) [place]  (1c) {1};
	\node at (-4,-2) [place] (3c) {3};
	\draw [thick,-] (2c.south) -- (1c.north);
	\draw [thick,-] (1c.south) -- (3c.north);
	\draw [thick,-] (3c) to [out=135,in=-135] (2c);
	\node at (-2,4) [place]  (1'c) {1};
	\node[label=$\mathcal{G}_1$] at (0,4) [place]   (2'c)  {2};
	\node at (2,4) [place]   (3'c) {3};
	\draw [thick,-] (1'c.east) -- (2'c.west);
	\draw [thick,-] (2'c.east) -- (3'c.west);
	\node at (-2,0) [place]  (11c) {1,1};
	\node at (0,0) [place]   (21c) {2,1};
	\node at (2,0) [place]   (31c) {3,1};
	\node at (-2,2) [place]  (12c) {1,2};
	\node[label=$\mathcal{G}_1 * \mathcal{G}_2$] at (0,2) [place]   (22c) {2,2};
	\node at (2,2) [place]   (32c) {3,2};
	\node at (-2,-2) [place] (13c) {1,3};
	\node at (0,-2) [place]  (23c) {2,3};
	\node at (2,-2) [place]  (33c) {3,3};
	\draw [thick,-] (11c.south) -- (13c.north);
	\draw [thick,-] (12c.south) -- (11c.north);
	\draw [thick,-] (21c.south) -- (23c.north);
	\draw [thick,-] (22c.south) -- (21c.north);
	\draw [thick,-] (31c.south) -- (33c.north);
	\draw [thick,-] (32c.south) -- (31c.north);
	\draw [thick,-] (12c.east) -- (22c.west);
	\draw [thick,-] (11c.east) -- (21c.west);
	\draw [thick,-] (13c.east) -- (23c.west);
	\draw [thick,-] (22c.east) -- (32c.west);
	\draw [thick,-] (21c.east) -- (31c.west);
	\draw [thick,-] (23c.east) -- (33c.west);
	\draw [thick,-] (12c) to [out=-45,in=135]  (21c);
	\draw [thick,-] (11c) to [out=-45,in=135]  (23c);
	\draw [thick,-] (12c) to [out=-55,in=125]  (23c);
	\draw [thick,-] (13c) to [out=55,in=-125]  (22c);
	\draw [thick,-] (22c) to [out=-55,in=125]  (33c);
	\draw [thick,-] (23c) to [out=55,in=-125]  (32c);
	\draw [thick,-] (13c) to [out=45,in=-135]  (21c);
	\draw [thick,-] (11c) to [out=45,in=-135]  (22c);
	\draw [thick,-] (21c) to [out=45,in=-135]  (32c);
	\draw [thick,-] (23c) to [out=45,in=-135]  (31c);
	\draw [thick,-] (22c) to [out=-45,in=135]  (31c);
	\draw [thick,-] (21c) to [out=-45,in=135]  (33c);
	\draw [thick,-] (12c) to [out=-135,in=135]  (13c);
	\draw [thick,-] (22c) to [out=-115,in=115]  (23c);
	\draw [thick,-] (32c) to [out=-55,in=55]  (33c);
	\draw [thick,-] (13c) to [out=35,in=-105] (32c);
	\draw [thick,-] (13c) to [out=30,in=-145] (31c);
	\draw [thick,-] (11c) to [out=-35,in=160] (33c);
	\draw [thick,-] (11c) to [out=35,in=-155] (32c);
	\draw [thick,-] (12c) to [out=-35,in=155] (31c);
	\draw [thick,-] (12c) to [out=-75,in=155] (33c);
	\end{tikzpicture}
	\caption{Disjunctive product of $\mathcal{G}_1$ and $\mathcal{G}_2$.}
	\label{figdisjunct}
\end{figure}
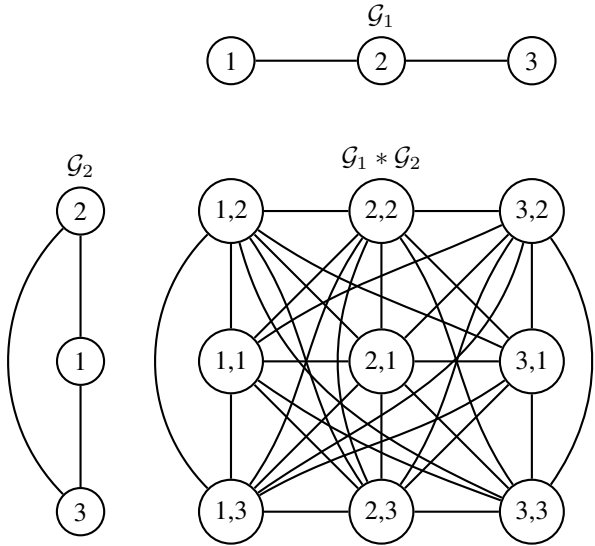

An example of the lexicographic product of two graphs $\mathcal{G}_1$ and $\mathcal{G}_2$ is shown in Figure \ref{figlexico}. For the same $\mathcal{G}_1$ and $\mathcal{G}_2$, the disjunctive product is shown in Figure \ref{figdisjunct}.

For any unicast ICP (either single-sender or two-sender), the knowledge of side-information and demands of all the receivers is represented by the side-information digraph $\mathcal{D}=(\mathcal{V}(\mathcal{D}),\mathcal{E}(\mathcal{D}))$, where the vertex set is given by $\mathcal{V}(\mathcal{D})=\{v_{1},\cdots,v_{m}\}$. The vertex $v_{i}$ represents the $i$th receiver which demands the message ${\bf{x}}_i$. Due to the one-to-one relationship between the $i$th receiver and ${\bf{x}}_i$, $v_{i}$ also represents ${\bf{x}}_i$. Hence, we refer to  $v_i$ as the $i$th message, the $i$th receiver and the  $i$th vertex  interchangeably. The edge set is given by $\mathcal{E}(\mathcal{D})=\{(v_{i},v_{j}): {\bf{x}}_{j} \in \mathcal{K}_{i}, i, j \in [m]\}$. 
The message sets $\mathcal{P}_{1} = \mathcal{M}_{1} \setminus \mathcal{M}_{2}$ and $\mathcal{P}_{2} = \mathcal{M}_{2} \setminus \mathcal{M}_{1}$ contain the messsages available only with $S_{1}$ and $S_{2}$ respectively.  $\mathcal{P}_{3} = \mathcal{M}_{1} \cap \mathcal{M}_{2}$ is the set of messages available with both the senders. Let $m_i=|\mathcal{P}_{i}|, i \in \{1,2,3\}$. Let $\mathcal{P}=(\mathcal{P}_{1},\mathcal{P}_{2},\mathcal{P}_{3})$. Any TUICP $\mathcal{I}$ can be described in terms of the two tuple $(\mathcal{D},\mathcal{P})$, as $\mathcal{I}(\mathcal{D},\mathcal{P})$. The optimal broadcast rates $\beta$ and $\beta_{t}$ of any TUICP $\mathcal{I}(\mathcal{D},\mathcal{P})$ are denoted by $\beta(\mathcal{D},\mathcal{P})$ and $\beta_{t}(\mathcal{D},\mathcal{P})$ respectively. Similarly, an achievable broadcast rate $p_t$ is denoted by $p_t(\mathcal{D},\mathcal{P})$. For a single-sender unicast ICP with side-information digraph $\mathcal{D}$, the $\beta$ and $\beta_{t}$ are denoted by $\beta(\mathcal{D})$ and $\beta_{t}(\mathcal{D})$ respectively.

The TUICP has been analyzed using three disjoint sub-digraphs of the side-information digraph (equivalently three sub-problems) induced by the three disjoint vertex sets respectively \cite{CTLO}. Let $\mathcal{D}_s$ be the sub-digraph of $\mathcal{D}$, induced by the vertices $\{v_j: {\bf{x}}_j \in \mathcal{P}_{s}, j \in [m]\}$, where $s \in \{1,2,3\}$. If there exists an edge from some vertex in $\mathcal{V}(\mathcal{D}_{i})$ to some vertex in $\mathcal{V}(\mathcal{D}_{j})$, in the side-information digraph $\mathcal{D}$, $i,j \in \{1,2,3\}, i \neq j$, then we say that there is an interaction from $\mathcal{D}_{i}$ to $\mathcal{D}_{j}$, and denote it by $\mathcal{D}_{i} \rightarrow \mathcal{D}_{j}$. We say that the interaction $\mathcal{D}_{i} \rightarrow \mathcal{D}_{j}$ is fully-participated, if there are edges from every vertex in $\mathcal{V}(\mathcal{D}_{i})$ to every vertex in $\mathcal{V}(\mathcal{D}_{j})$. Otherwise, it is said to be a partially-participated interaction. We say that the TUICP has fully-participated interactions if all the existing interactions are fully-participated interactions. For a given two-sender problem, we define a digraph called the interaction digraph, which captures the type of interactions between the sub-digraphs of the side-information digraph. 

\begin{defn}[Interaction digraph]
	For a given TUICP $\mathcal{I}(\mathcal{D},\mathcal{P})$, the digraph $\mathcal{H}$ with  $\mathcal{V}(\mathcal{H}) = \{1,2,3\}$ and $\mathcal{E}(\mathcal{H})=\{(i,j) | \mathcal{D}_{i} \rightarrow \mathcal{D}_{j}, i \neq j, i,j \in \{1,2,3\}\}$, is defined as the interaction digraph of the side-information digraph $\mathcal{D}$.
	\label{definterdigraph}
\end{defn}

Note that a given side-information digraph can correspond to different interaction digraphs based on the choice of the message tuple $\mathcal{P}$. The edges $(i,j)$ and $(j,i)$ in any interaction digraph are denoted by a single edge with arrows at both the ends, $i,j \in \{1,2,3\}$. There are 64 possibilities for the digraph $\mathcal{H}$ as shown in Figure \ref{interenum}, which were enlisted  and classified in \cite{CTLO}. The number written below each interaction digraph in the figure is used as the subscript to denote the specific interaction digraph. The side-information digraph $\mathcal{D}$ describing a given two-sender problem with the interaction digraph $\mathcal{H}_{k}$ is denoted by $\mathcal{D}^{k}$, $k \in \{1,2,\cdots,64\}$. For any TUICP $\mathcal{I}(\mathcal{D}^k,\mathcal{P})$, the corresponding sub-digraphs $\mathcal{D}_i$, $i \in \{1,2,3\}$, are denoted as $\mathcal{D}_{i}^{k,\mathcal{P}}$. Any TUICP $\mathcal{I}(\mathcal{D}^k,\mathcal{P})$ is analyzed using the three single-sender unicast ICPs with the  side-information digraphs $\mathcal{D}_{i}^{k,\mathcal{P}}$, $i \in \{1,2,3\}$. Note that all the possible interaction digraphs are classified into two cases broadly: Case I and Case II. Case I consists of acyclic interaction digraphs (i.e., with no cycles). Case II is further classified into five sub-cases as shown in Figure \ref{interenum}. We illustrate the above definitions using an example.    
\begin{figure}[!htbp]
	\begin{center}
		\begin{tikzpicture}
		[place/.style={circle,draw=black!100,thick}]
		\node at (-5.5,.2) [place] (1c) {1};
		\node at (-6.7,-.7) [place] (5c) {5};
		\node at (-6,-1.8) [place] (4c) {4};
		\node at (-4.7,-1.8) [place] (3c) {3};
		\node at (-4.2,-.7) [place] (2c) {2};
		\node at (-3.6,-1.8) [place] (1h) {1}; 	 	
		\node at (-2.1,-1.8) [place] (2h) {2};
		\node at (-2.8,-.7) [place] (3h) {3};	
		\node at (-0.5,-.2) [place] (1'h) {1}; 	 	
		\node at (1,-.2) [place] (2'h) {2};
		\node at (1,-1.1) [place] (3'h) {3};    		     	
		\node at (-0.5,-1.1) [place] (4'h) {4};
		\node at (0.2,-1.9) [place] (5'h) {5}; 		     	
		\draw (-5.4,-2.5) node {$\mathcal{D}$};
		\draw (-2.9,-2.5) node {$\mathcal{H}$};
		\draw (-1.3,-.2) node {$\mathcal{D}_1$};
		\draw (-1.3,-1.1) node {$\mathcal{D}_2$};     		  
		\draw (-0.9,-2) node {$\mathcal{D}_{3}$};
		\draw [thick,->] (2h) to [out=110,in=-45] (3h);     		     	
		\draw [thick,<->] (3h) to [out=-140,in=65] (1h);    		     	
		\draw [thick,->] (1'h) to [out=20,in=160] (2'h);     		    
		\draw [thick,->] (2'h) to [out=-160,in=-20] (1'h);
		\draw [thick,->] (3'h) to [out=160,in=20] (4'h);     		     	
		\draw [thick,->] (4'h) to [out=-20,in=-160] (3'h);     		     	
		\draw [thick,->] (1c) to [out=-5,in=120] (2c);     	
		\draw [thick,->] (3c) to [out=160,in=20] (4c);    		     	
		\draw [thick,->] (4c) to [out=-20,in=-160] (3c);
		\draw [thick,->] (3c) to [out=135,in=-35] (5c);     		     	
		\draw [thick,->] (2c) to [out=160,in=-45] (1c);    	
		\draw [thick,->] (1c) to [out=-170,in=65] (5c);  
		\draw [thick,->] (5c) to [out=30,in=-130] (1c);  
		\draw [thick,->] (5c) to [out=0,in=180] (2c);	     	
		\end{tikzpicture}
		\caption{Example to illustrate the interaction digraph and the three sub-digraphs of a given side-information digraph of the two-sender problem given in Example \ref{exampp1}.}
		\label{examp1}
	\end{center}
\end{figure}
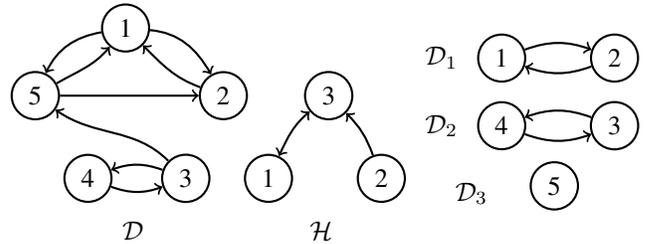
\begin{figure*}[!htbp]
	\begin{center}
		\includegraphics[width=41pc]{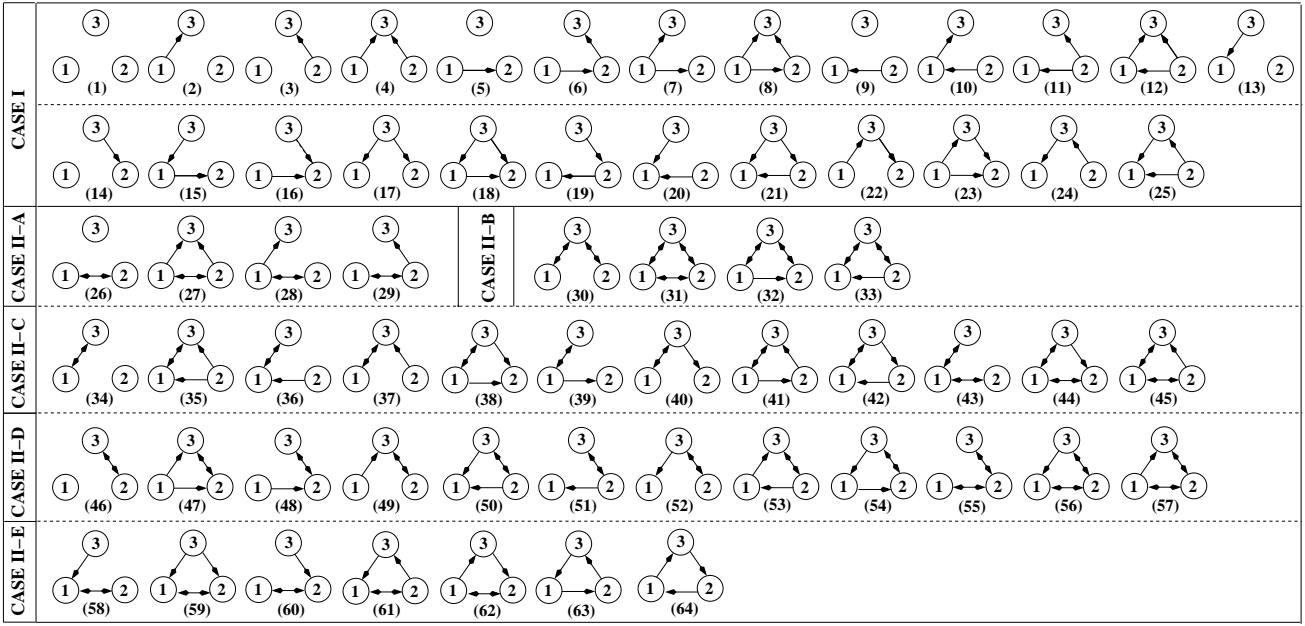}
		\caption{Enumeration of all the possible interactions between the sub-digraphs $\mathcal{D}_1$, $\mathcal{D}_2$, and $\mathcal{D}_3$, denoted by the interaction digraph $\mathcal{H}$.}
		\label{interenum}
	\end{center}
	\hrule
\end{figure*}
\begin{exmp}
	Consider the TUICP with $m=5$ messages, where the $i$th receiver demands $i$th message ${\bf{x}}_{i}$. Sender $\mathcal{S}_1$ has $\mathcal{M}_1=\{{\bf{x}}_1,{\bf{x}}_2,{\bf{x}}_5\}$. Sender $\mathcal{S}_2$ has $\mathcal{M}_2=\{{\bf{x}}_3,{\bf{x}}_4,{\bf{x}}_5\}$. Hence, $\mathcal{P}_1=\{{\bf{x}}_1,{\bf{x}}_2\}$, $\mathcal{P}_2=\{{\bf{x}}_3,{\bf{x}}_4\}$, and $\mathcal{P}_3=\{{\bf{x}}_5\}$. The side-information of all the receivers are given as follows: $\mathcal{K}_{1}=\{{\bf{x}}_{2},{\bf{x}}_{5}\}$, $\mathcal{K}_{2}=\{{\bf{x}}_{1}\}$, $\mathcal{K}_{3}=\{{\bf{x}}_{4},{\bf{x}}_{5}\}$, $\mathcal{K}_{4}=\{{\bf{x}}_{3}\}$, $\mathcal{K}_{5}=\{{\bf{x}}_{1},{\bf{x}}_{2}\}$. The side-information digraph  $\mathcal{D}$ and the corresponding interaction digraph  $\mathcal{H}$ are shown in Figure \ref{examp1}.  The vertex-induced sub-digraphs $\mathcal{D}_1$, $\mathcal{D}_2$, and $\mathcal{D}_3$ induced by the messages in $\mathcal{P}_1$, $\mathcal{P}_2$, and $\mathcal{P}_3$ respectively are also shown in the figure. Note that the interaction $\mathcal{D}_{3} \rightarrow \mathcal{D}_{1}$ is fully-participated. Others are partially-participated interactions. The interaction digraph shown in Figure \ref{examp1} is $\mathcal{H}_{37}$ as given in Figure \ref{interenum}. Hence, the side-information digraph $\mathcal{D}$ can also be denoted as $\mathcal{D}^{37}$.
	\label{exampp1}
\end{exmp}

\par The following notations are required for the construction of a two-sender index code from single-sender index codes. Let $\mathcal{C}_1$ and $\mathcal{C}_2$ be  two codewords of length $l_1$ and $l_2$ respectively. $\mathcal{C}_1 \oplus \mathcal{C}_2$ denotes the bit-wise XOR of $\mathcal{C}_1$ and $\mathcal{C}_2$ after zero-padding the shorter message at the least significant positions to match the length of the longer message. The resulting length of the codeword is  $max(l_1,l_2)$. For example, if $\mathcal{C}_1=1010$, and $\mathcal{C}_2=110$, then $\mathcal{C}_1 \oplus \mathcal{C}_2 = 0110$.  $\mathcal{C}[a:b]$ denotes the vector obtained by picking the bits from bit position $a$ to bit position $b$, starting from the most significant position of the codeword $\mathcal{C}$, with $a,b \in [l]$, $l$ being the length of $\mathcal{C}$. For example $\mathcal{C}_1[2:4]=010$.

\section{Confusion Graphs and the Two-sender graph coloring}
\par In this section, we review confusion graphs and recapitulate some results on the two-sender graph coloring of confusion graphs provided in \cite{CTLO}. We also provide some definitions which are used to describe the symmetries of the confusion graph. Then, we state and prove a lemma related to these symmetries, which is used to establish the main results in this paper.

\par Consider a unicast ICP (single-sender or two-sender) described by a side-information digraph $\mathcal{D}$ with $m$ messages. Let ${\bf{x}} = ({\bf{u}}_1,...,{\bf{u}}_m)$ and ${\bf{x}}' = ({\bf{v}}_1,...,{\bf{v}}_m)$ be two tuples of realizations of $m$ messages, where ${\bf{u}}_i,{\bf{v}}_i \in \mathbb{F}_{2}^{t}, \forall i \in [m]$. The tuples ${\bf{x}}$ and ${\bf{x'}}$  are said to be confusable at the $i$th receiver, if ${\bf{u}}_i \neq {\bf{v}}_i$ and ${\bf{u}}_j = {\bf{v}}_j$ for all $j$ such that ${\bf{x}}_j \in \mathcal{K}_i$. Two tuples are said to be confusable if they are confusable at some reciever. Confusion at a receiver refers to existence of confusable tuples at the receiver. In index coding, two tuples of realizations of $m$ messages  that are confusable cannot be encoded to the same codeword as one of the receivers cannot decode the demanded message succesfully using the broadcasted codeword and its side-information. The confusion graph is defined as follows.
\begin{defn}[Confusion graph, \cite{CTLO}] 
	The confusion graph of a side-information digraph $\mathcal{D}$ with $m$ vertices and $t$-bit messages is an undirected graph, denoted by $\Gamma_{t}(\mathcal{D})=(\mathcal{V}(\Gamma_{t}(\mathcal{D})),\mathcal{E}(\Gamma_{t}(\mathcal{D})))$, where $\mathcal{V}(\Gamma_{t}(\mathcal{D})) = \{ {\bf{x}} : {\bf{x}} \in \mathbb{F}_{2}^{mt} \}$ and $\mathcal{E}(\Gamma_{t}(\mathcal{D})) = \{ ({\bf{x}},{\bf{x}}') :  {\bf{x}}$ and ${\bf{x}}'$ are confusable$\}$.
\end{defn}

We require the following notations and definitions to state our results in the following section. For any TUICP with the side-information digraph $\mathcal{D}^{k}$, $k \in [64]$, and any message set tuple $\mathcal{P}$ with $t$-bit messages for any finite $t$, we have the following definitions.

\begin{defn}
	$\mathcal{D}_{u \circ v}^{k,\mathcal{P}}$ denotes the side-information digraph whose confusion graph $\Gamma_t(\mathcal{D}_{u \circ v}^{k,\mathcal{P}})$ is given by  $\Gamma_t(\mathcal{D}_{u}^{k,\mathcal{P}}) \circ  \Gamma_t(\mathcal{D}_{v}^{k,\mathcal{P}}) $, $u \neq v, u,v \in \{1,2,3\}$ (where the  lexicographic graph product denoted using $``\circ"$ is given in Definition \ref{deflexi}).
	\label{deflexi1}
\end{defn}

\begin{defn}
	$\mathcal{D}_{u*v}^{k,\mathcal{P}}$ denotes the side-information digraph whose confusion graph $\Gamma_t(\mathcal{D}_{u*v}^{k,\mathcal{P}})$ is given by  $\Gamma_t(\mathcal{D}_{u}^{k,\mathcal{P}}) * \Gamma_t(\mathcal{D}_{v}^{k,\mathcal{P}})$, $u \neq v, u,v \in \{1,2,3\}$ (where the disjunctive graph product denoted using $``*"$ is given in Definition \ref{defdisj}). 
	\label{defdisj1}
\end{defn}

\par We use the following notation used in \cite{CTLO}, in the context of confusion graphs. Each realization of the bits of concatenated messages belonging to $\mathcal{P}_1$, $\mathcal{P}_2$, and $\mathcal{P}_3$, (i.e., each element of $\mathbb{F}_2^{tm_1}$, $\mathbb{F}_2^{tm_2}$ and $\mathbb{F}_2^{tm_3}$ respectively), is represented by unique tuples  ${\bf{b}}_{\mathcal{P}_{1}}^{i}$,  ${\bf{b}}_{\mathcal{P}_{2}}^{j}$, and ${\bf{b}}_{\mathcal{P}_{3}}^{k}$  respectively. Superscripts $i,i' \in [2^{tm_1}],$ $j,j' \in [2^{tm_2}],$ and $k,k' \in [2^{tm_3}]$ are used to represent possible realizations of concatenation of all the messages belonging to $\mathcal{P}_1, \mathcal{P}_2$, and $\mathcal{P}_3$ of $tm_1$, $tm_2$, and $tm_3$ bits respectively. Each message tuple $({\bf{x}}_1,...,{\bf{x}}_m)$ can be uniquely written as $({\bf{b}}_{\mathcal{P}_{1}}^{i},{\bf{b}}_{\mathcal{P}_{2}}^{j},{\bf{b}}_{\mathcal{P}_{3}}^{k})$ for some $i,j$, and $k$. Hence, each vertex of the confusion graph can be labelled by a unique tuple $({\bf{b}}_{\mathcal{P}_{1}}^{i},{\bf{b}}_{\mathcal{P}_{2}}^{j},{\bf{b}}_{\mathcal{P}_{3}}^{k})$.
\par Consider a valid coloring of the confusion graph $\Gamma_{t}(\mathcal{D})$ with a set of colors $\mathcal{J}$. This results in $|\mathcal{J}|$ sets of vertices, such that all the vertices in a given set are colored with a unique color. Each set of vertices is independent and can be coded into the same codeword, as no pair of vertices in the given set are confusable. Hence, sending a codeword is equivalent to sending the identity of a color. As $\chi(\Gamma_{t}(\mathcal{D}))$ is the minimum number of colors required, the optimal codelength is $\lceil \log_2 \chi(\Gamma_{t}(\mathcal{D}))  \rceil$ bits. The classical graph coloring of the confusion graph may not yield the optimal codelength for the two-sender unicast ICP, as there is a constraint on the  coloring due to the non-availability of some messages at one of the senders. To account for the encoding done by the two senders, two-sender graph coloring had been introduced in \cite{CTLO}.
\begin{defn}[Two-sender graph coloring of $\Gamma_{t}(\mathcal{D})$, \cite{CTLO}] 	
	Let two onto functions $J_1: \mathbb{F}_{2}^{tm_1} \times \mathbb{F}_{2}^{tm_3} \rightarrow \mathcal{J}_1$ and $J_2 : \mathbb{F}_{2}^{tm_2} \times \mathbb{F}_{2}^{tm_3} \rightarrow \mathcal{J}_2$ be the coloring functions carried out by senders $\mathcal{S}_1$ and $\mathcal{S}_2$ respectively. A proper two-sender graph coloring of $\Gamma_{t}(\mathcal{D})$ is an onto function $J_0: \mathbb{F}_{2}^{tm_1} \times \mathbb{F}_{2}^{tm_2} \times \mathbb{F}_{2}^{tm_3} \rightarrow \mathcal{J}_1 \times \mathcal{J}_2$ where $J_o(({\bf{b}}_{\mathcal{P}_{1}}^{i},{\bf{b}}_{\mathcal{P}_{2}}^{j},{\bf{b}}_{\mathcal{P}_{3}}^{k}))=(J_{1}({\bf{b}}_{\mathcal{P}_{1}}^{i},{\bf{b}}_{\mathcal{P}_{3}}^{k}),J_2({\bf{b}}_{\mathcal{P}_{2}}^{j},{\bf{b}}_{\mathcal{P}_{3}}^{k}))$ such that if $({\bf{b}}_{\mathcal{P}_{1}}^{i},{\bf{b}}_{\mathcal{P}_{2}}^{j},{\bf{b}}_{\mathcal{P}_{3}}^{k})$ and $({\bf{b}}_{\mathcal{P}_{1}}^{i'},{\bf{b}}_{\mathcal{P}_{2}}^{j'},{\bf{b}}_{\mathcal{P}_{3}}^{k'})$ are adjacent vertices of $\Gamma_{t}(\mathcal{D})$,
	then $J_o(({\bf{b}}_{\mathcal{P}_{1}}^{i},{\bf{b}}_{\mathcal{P}_{2}}^{j},{\bf{b}}_{\mathcal{P}_{3}}^{k})) \neq J_o(({\bf{b}}_{\mathcal{P}_{1}}^{i'},{\bf{b}}_{\mathcal{P}_{2}}^{j'},{\bf{b}}_{\mathcal{P}_{3}}^{k'}))$.
\end{defn}	
\par Note that the two ordered pairs of colors given by $(c_1,c_2)$ and $(c_{1}',c_{2}')$, where $c_i,c_{i}' \in \mathcal{J}_{i}$, with $i \in \{1,2\}$ are said to be different iff $c_1 \neq c_{1}'$ or $c_2 \neq c_{2}'$  or both. We recapitulate some basic results on the two-sender graph coloring stated as Lemmas 1 to 4 in \cite{CTLO}. These lemmas are used in coloring the confusion graph according to the two-sender graph coloring.
\begin{lem}[Lemma 1, \cite{CTLO}]
	For any two vertices $({\bf{b}}_{\mathcal{P}_{1}}^{i},{\bf{b}}_{\mathcal{P}_{2}}^{j},{\bf{b}}_{\mathcal{P}_{3}}^{k})$ and $({\bf{b}}_{\mathcal{P}_{1}}^{i'},{\bf{b}}_{\mathcal{P}_{2}}^{j},{\bf{b}}_{\mathcal{P}_{3}}^{k})$ in  $\Gamma_{t}(\mathcal{D})$ which are confusable, if  $J_o(({\bf{b}}_{\mathcal{P}_{1}}^{i},{\bf{b}}_{\mathcal{P}_{2}}^{j},{\bf{b}}_{\mathcal{P}_{3}}^{k})) = (c_1,c_2)$  and $J_o(({\bf{b}}_{\mathcal{P}_{1}}^{i'},{\bf{b}}_{\mathcal{P}_{2}}^{j},{\bf{b}}_{\mathcal{P}_{3}}^{k})) = (c_{1}',c_{2}')$, then we must have $c_1 \neq c_{1}'$ and $c_2 = c_{2}'$ for some $c_1,c_{1}' \in \mathcal{J}_{1}$ and $c_2,c_{2}' \in \mathcal{J}_{2}$.
	\label{lemThapaColor1}
\end{lem}
\begin{lem}[Lemma 2, \cite{CTLO}]
	For any two vertices $({\bf{b}}_{\mathcal{P}_{1}}^{i},{\bf{b}}_{\mathcal{P}_{2}}^{j},{\bf{b}}_{\mathcal{P}_{3}}^{k})$ and $({\bf{b}}_{\mathcal{P}_{1}}^{i},{\bf{b}}_{\mathcal{P}_{2}}^{j'},{\bf{b}}_{\mathcal{P}_{3}}^{k})$ in  $\Gamma_{t}(\mathcal{D})$ which are confusable, if  $J_o(({\bf{b}}_{\mathcal{P}_{1}}^{i},{\bf{b}}_{\mathcal{P}_{2}}^{j},{\bf{b}}_{\mathcal{P}_{3}}^{k})) = (c_1,c_2)$  and $J_o(({\bf{b}}_{\mathcal{P}_{1}}^{i},{\bf{b}}_{\mathcal{P}_{2}}^{j'},{\bf{b}}_{\mathcal{P}_{3}}^{k})) = (c_{1}',c_{2}')$, then we must have $c_1 = c_{1}'$ and $c_2 \neq c_{2}'$ for some $c_1,c_{1}' \in \mathcal{J}_{1}$ and $c_2,c_{2}' \in \mathcal{J}_{2}$.
	\label{lemThapaColor2}
\end{lem}
\begin{lem}[Lemma 3, \cite{CTLO}]
	For any two vertices $({\bf{b}}_{\mathcal{P}_{1}}^{i},{\bf{b}}_{\mathcal{P}_{2}}^{j},{\bf{b}}_{\mathcal{P}_{3}}^{k})$ and $({\bf{b}}_{\mathcal{P}_{1}}^{i'},{\bf{b}}_{\mathcal{P}_{2}}^{j'},{\bf{b}}_{\mathcal{P}_{3}}^{k})$ in  $\Gamma_{t}(\mathcal{D})$ which are confusable due to some vertices in $\mathcal{D}_{1}$ and $\mathcal{D}_{2}$, if  $J_o(({\bf{b}}_{\mathcal{P}_{1}}^{i},{\bf{b}}_{\mathcal{P}_{2}}^{j},{\bf{b}}_{\mathcal{P}_{3}}^{k})) = (c_1,c_2)$  and $J_o(({\bf{b}}_{\mathcal{P}_{1}}^{i'},{\bf{b}}_{\mathcal{P}_{2}}^{j'},{\bf{b}}_{\mathcal{P}_{3}}^{k})) = (c_{1}',c_{2}')$, then we must have $c_1 \neq c_{1}'$ and $c_2 \neq c_{2}'$ for some $c_1,c_{1}' \in \mathcal{J}_{1}$ and $c_2,c_{2}' \in \mathcal{J}_{2}$.
	\label{lemThapaColor3}
\end{lem}
\begin{lem}[Lemma 4, \cite{CTLO}]
	For any two vertices $({\bf{b}}_{\mathcal{P}_{1}}^{i},{\bf{b}}_{\mathcal{P}_{2}}^{j},{\bf{b}}_{\mathcal{P}_{3}}^{k})$ and $({\bf{b}}_{\mathcal{P}_{1}}^{i},{\bf{b}}_{\mathcal{P}_{2}}^{j},{\bf{b}}_{\mathcal{P}_{3}}^{k'})$ in  $\Gamma_{t}(\mathcal{D})$ which are confusable, if  $J_o(({\bf{b}}_{\mathcal{P}_{1}}^{i},{\bf{b}}_{\mathcal{P}_{2}}^{j},{\bf{b}}_{\mathcal{P}_{3}}^{k})) = (c_1,c_2)$  and $J_o(({\bf{b}}_{\mathcal{P}_{1}}^{i},{\bf{b}}_{\mathcal{P}_{2}}^{j},{\bf{b}}_{\mathcal{P}_{3}}^{k'})) = (c_{1}',c_{2}')$, then we must have either $c_1 \neq c_{1}'$, or $c_2 \neq c_{2}'$, or both, for some $c_1,c_{1}' \in \mathcal{J}_{1}$ and $c_2,c_{2}' \in \mathcal{J}_{2}$.
	\label{lemThapaColor4}
\end{lem}
\par The optimal broadcast rate for the TUICP with $t$-bit messages for every finite $t$ is given by Theorem 2 in \cite{CTLO}.
\begin{lem}[Theorem 2, \cite{CTLO}]
	\begin{equation} \beta_{t}(\mathcal{D},\mathcal{P})=\underset{\mathcal{J}_1,\mathcal{J}_2}{min}{ \frac{ \lceil \log_{2}{|\mathcal{J}_{1}|} \rceil + \lceil \log_{2}{|\mathcal{J}_{2}|} \rceil}{t} }
	\end{equation} 
	\label{lemtwocolorrate}	
\end{lem}
We illustrate the two-sender graph coloring of the confusion graph using an example. 
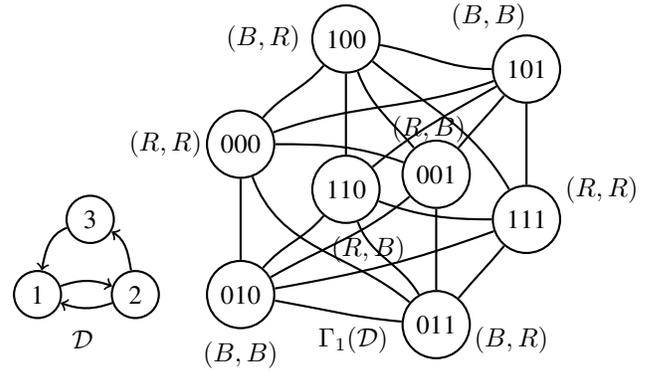
\begin{figure}[!htbp]
	\begin{center}
		\begin{tikzpicture}
		[place/.style={circle,draw=black!100,thick}]
		\node at (-6.5,-2.2) [place] (1'h) {1}; 	 	
		\node at (-5.2,-2.2) [place] (2'h) {2};
		\node at (-5.8,-1.2) [place] (3'h) {3};    		     	  	
		\draw (-5.9,-2.8) node {$\mathcal{D}$};     	
		\draw [thick,->] (1'h) to [out=20,in=160] (2'h);     		    
		\draw [thick,->] (2'h) to [out=-160,in=-20] (1'h);
		\draw [thick,->] (3'h) to [out=-160,in=80] (1'h);  
		\draw [thick,->] (2'h) to [out=100,in=-30] (3'h); 
		\node at (-3.8,-0.2) [place] (0'h) {000}; 	 	
		\node at (-3.8,-2.2) [place] (2'h) {010};
		\node at (-2.4,1.2) [place] (4'h) {100};
		\node at (-2.4,-.8) [place] (6'h) {110}; 	 	
		\node at (-1.2,-0.6) [place] (1'h) {001};
		\node at (-1.2,-2.6) [place] (3'h) {011}; 		
		\node at (0,0.8) [place] (5'h) {101};
		\node at (0,-1.2) [place] (7'h) {111};     
		\draw [thick,-] (2'h) to [out=-10,in=175] (3'h);			
		\draw [thick,-] (6'h) to [out=-20,in=180] (7'h);	
		\draw [thick,-] (0'h) to [out=0,in=160] (1'h);			
		\draw [thick,-] (4'h) to [out=-10,in=-180] (5'h);
		\draw [thick,-] (0'h) to [out=50,in=-130] (4'h);			
		\draw [thick,-] (1'h) to [out=50,in=-130] (5'h);
		\draw [thick,-] (2'h) to [out=50,in=-130] (6'h);			
		\draw [thick,-] (3'h) to [out=50,in=-130] (7'h);
		\draw [thick,-] (0'h) to [out=-90,in=90] (2'h);			
		\draw [thick,-] (4'h) to [out=-90,in=90] (6'h);
		\draw [thick,-] (5'h) to [out=-90,in=90] (7'h);			
		\draw [thick,-] (1'h) to [out=-90,in=90] (3'h);
		\draw [thick,-] (2'h) to [out=10,in=-160] (7'h);			
		\draw [thick,-] (0'h) to [out=20,in=-160] (5'h);
		\draw [thick,-] (6'h) to [out=-70,in=120] (3'h);			
		\draw [thick,-] (4'h) to [out=-70,in=130] (1'h);
		\draw [thick,-] (2'h) to [out=30,in=-140] (1'h);			
		\draw [thick,-] (6'h) to [out=40,in=-150] (5'h);
		\draw [thick,-] (0'h) to [out=-70,in=140] (3'h);			
		\draw [thick,-] (4'h) to [out=-40,in=120] (7'h);
		\draw (-2.3,-2.8) node {$\Gamma_1(\mathcal{D})$};
		\draw (-4.8,-0.2) node {$(R,R)$};
		\draw (-1.3,0) node {$(R,B)$};  
		\draw (-3.8,-3) node {$(B,B)$};
		\draw (-0.2,-2.8) node {$(B,R)$};       
		\draw (-3.5,1.2) node {$(B,R)$};
		\draw (-.5,1.5) node {$(B,B)$};  
		\draw (-2.1,-1.6) node {$(R,B)$};
		\draw (1,-0.8) node {$(R,R)$}; 
		\end{tikzpicture}
		
		\caption{Side-information digraph and two-sender graph coloring of its confusion graph for the two-sender problem given in Example \ref{exampp2}.}
		\label{con_graph}
	\end{center}
\end{figure}
\begin{exmp}
	Consider the following TUICP with $t=1$ and $m=3$ messages. $S_1$ has $\mathcal{M}_{1}=\{{\bf{x}}_1,{\bf{x}}_3\}$ and $S_2$ has $\mathcal{M}_{2}=\{{\bf{x}}_2,{\bf{x}}_3\}$. Hence, $\mathcal{P}_1={\bf{x}}_1$, $\mathcal{P}_2={\bf{x}}_2$, and $\mathcal{P}_3={\bf{x}}_3$. The side-information of the  receivers are as follows: $\mathcal{K}_1={\bf{x}}_2,$ $\mathcal{K}_2=\{{\bf{x}}_1,{\bf{x}}_3\}$, and  $\mathcal{K}_3={\bf{x}}_1$. Hence, we have $\mathcal{V}(\mathcal{D}_i)={\bf{x}}_i,$ $i \in \{1,2,3\}$. The side-information digraph and its confusion graph are shown in Figure \ref{con_graph}. The confusion graph $\Gamma_{1}(\mathcal{D})$ has $2^m=8$ vertices representing all possible binary tuples $({\bf{b}}_{\mathcal{P}_{1}}^{i},{\bf{b}}_{\mathcal{P}_{2}}^{j},{\bf{b}}_{\mathcal{P}_{3}}^{k})$, $i,j,k \in \{1,2\}$, of length three. Edges are drawn between every two confusable tuples. For example, there is an edge between $(0,1,0)$ and $(1,1,1)$ due to confusion at receiver $1$. After the construction of the confusion graph, all the vertices are colored by each sender. In the ordered pair of colors, the first color is associated with $\mathcal{S}_1$ and the second color is associated with $\mathcal{S}_2$. Color RED is denoted as R and BLUE is denoted as B in Figure \ref{con_graph}. Coloring is done based on Lemmas \ref{lemThapaColor1} to \ref{lemThapaColor4}. Hence, if $S_1$ colors $(0,1,0)$ with BLUE color, it must color $(1,1,1)$ with another color, say RED. Similarly, we can color other vertices using the two-sender graph coloring. It can be easily verified that only two colors are required at each sender to color the confusion graph. The two-sender coloring shown in Figure \ref{con_graph} can be easily verified to be a valid two-sender coloring.  Hence, $\mathcal{J}_1=\mathcal{J}_2=\{$RED,BLUE$\}$. Assuming a map from the colors to binary bits that maps RED to $1$ and BLUE to $0$, the tuple $(0,0,0)$ can be mapped to the codeword $11$, the tuple $(0,0,1)$ can be mapped to the codeword $10$, and so on. Thus the two-sender index code consists of codewords given by $\{00,01,10,11\}$. The first bit of the codeword is sent by $\mathcal{S}_1$, and the second bit is sent by $\mathcal{S}_2$. Thus, $\beta_{t}(\mathcal{D},\mathcal{P}) \leq 2$ for any $t \geq 1$. As each sender has a single message which is not present with the other sender, each of them must atleast send one bit. Thus,  $\beta_{t}(\mathcal{D},\mathcal{P}) \geq 2$. Hence, $\beta_{t}(\mathcal{D},\mathcal{P}) = 2$.
	\label{exampp2}
\end{exmp} 	 
To exploit the symmetries of confusion graphs and facilitate the two-sender graph coloring, the vertices of $\Gamma_{t}(\mathcal{D})$ can be grouped in different ways. Let $\mathcal{B}_{{\bf{b}}_{\mathcal{P}_{1}}^{i}} \triangleq \{({\bf{b}}_{\mathcal{P}_{1}}^{i},{\bf{b}}_{\mathcal{P}_{2}}^{j},{\bf{b}}_{\mathcal{P}_{3}}^{k}) :$ for some fixed ${\bf{b}}_{\mathcal{P}_{1}}^{i}\}$ with cardinality $2^{tm_2} \times 2^{tm_3}$. Similarly, $\mathcal{B}_{{\bf{b}}_{\mathcal{P}_{2}}^{j}}$ and $\mathcal{B}_{{\bf{b}}_{\mathcal{P}_{3}}^{k}}$ are also defined. The subgraph of $\Gamma_{t}(\mathcal{D})$ induced by the vertices belonging to  $\mathcal{B}_{{\bf{b}}_{\mathcal{P}_{3}}^{k}}$ is called the $k$th $K$-block. There are $2^{tm_{3}}$ $K$-blocks. Similarly, the subgraph of $\Gamma_{t}(\mathcal{D})$ induced by the vertices belonging to  $\mathcal{B}_{{\bf{b}}_{\mathcal{P}_{2}}^{j}}$ is called the $j$th $J$-block. There are $2^{tm_{2}}$ $J$-blocks. Similarly, $i$th $I$-block $\mathcal{B}_{{\bf{b}}_{\mathcal{P}_{1}}^{i}}$ is also defined. We define three types of inter-block edges.
\begin{defn}(Inter-block edges)
	An edge between two vertices, each belonging to a different $I$-block of $\Gamma_{t}(\mathcal{D})$ is called an inter-$I$-block edge. An edge between two vertices, each belonging to a different $J$-block of $\Gamma_{t}(\mathcal{D})$ is called an inter-$J$-block edge.
	An edge between two vertices, each belonging to a different $K$-block of $\Gamma_{t}(\mathcal{D})$ is called an inter-$K$-block edge.
\end{defn}
\par We require the following lemma to exploit the symmetry in the  two-sender graph coloring of the confusion graph. 

\begin{lem} All $I$-blocks in a given confusion graph are isomorphic to each other. Similarly, all $J$-blocks are isomorphic to each other, and all $K$-blocks are isomorphic to each other, in a given confusion graph. 
	\label{lemblockiso}
\end{lem}
\begin{proof}
	We prove the lemma for all $I$-blocks. The proof is similar for all $J$-blocks and all $K$-blocks.
	\par Every vertex in any $i$th $I$-block induced by the vertices in $\mathcal{B}_{{\bf{b}}_{\mathcal{P}_1}^{i}}$ has the same ${\bf{b}}_{\mathcal{P}_1}^{i}$ sub-label, $i \in [2^{tm_1}]$. Thus, any edge in any $i$th $I$-block is only due to confusion at the vertices (i.e, receivers) belonging to $\mathcal{V}(\mathcal{D}_{2} \cup \mathcal{D}_{3})$. Every $I$-block has $2^{tm_2} \times 2^{tm_3}$ vertices. If 
	there is an edge given by  $(({\bf{b}}_{\mathcal{P}_{1}}^{i},{\bf{b}}_{\mathcal{P}_{2}}^{j},{\bf{b}}_{\mathcal{P}_{3}}^{k}),({\bf{b}}_{\mathcal{P}_{1}}^{i},{\bf{b}}_{\mathcal{P}_{2}}^{j'},{\bf{b}}_{\mathcal{P}_{3}}^{k'}))$ in $i$th $I$-block induced by $\mathcal{B}_{{\bf{b}}_{\mathcal{P}_{1}}^{i}}$, then
	there is an edge given by  $(({\bf{b}}_{\mathcal{P}_{1}}^{i'},{\bf{b}}_{\mathcal{P}_{2}}^{j},{\bf{b}}_{\mathcal{P}_{3}}^{k}),({\bf{b}}_{\mathcal{P}_{1}}^{i'},{\bf{b}}_{\mathcal{P}_{2}}^{j'},{\bf{b}}_{\mathcal{P}_{3}}^{k'}))$ in $i'$th $I$-block induced by  $\mathcal{B}_{{\bf{b}}_{\mathcal{P}_{1}}^{i'}}$, $i \neq i'$ 
	and vice versa, as the confusion is only due to tuples $({\bf{b}}_{\mathcal{P}_{2}}^{j},{\bf{b}}_{\mathcal{P}_{3}}^{k})$ and $({\bf{b}}_{\mathcal{P}_{2}}^{j'},{\bf{b}}_{\mathcal{P}_{3}}^{k'})$, at some vertex  
	belonging to $\mathcal{V}(\mathcal{D}_{2} \cup \mathcal{D}_{3})$. Hence, all $I$-blocks are isomorphic to each other.
\end{proof}

\section{An achievable broadcast rate with finite length messages for some sub-cases of CASE I and all sub-cases of CASE II-E} 

In this section, we provide an  achievable broadcast rate with $t$-bit messages for any finite $t$, for some sub-cases of Case I with fully-participated interactions, using a valid two-sender graph coloring of the confusion graph. No non-trivial achievable broadcast rate  was known for these sub-cases with $t$-bit messages for any finite $t$. In particular, an achievable broadcast rate $\beta_{t}(\mathcal{D}^{k},\mathcal{P})$ is given for any TUICP with side-information digraph $\mathcal{D}^k$, $k \in \{16,18,20,21,23,25\}$, having fully-participated interactions between its sub-digraphs $\mathcal{D}_{i}^{k,\mathcal{P}}$, $i \in \{1,2,3\}$. 

We also provide an achievable broadcast rate with $t$-bit messages for any finite $t$, for all the  sub-cases of Case II-E with fully-participated interactions, using a code-construction based on the optimal codes for the single-sender sub-problems. This provides a tighter upper bound on $\beta_{t}(\mathcal{D}^{k},\mathcal{P})$, $k \in \{58,59,\cdots,64\}$, when compared to that given in \cite{CTLO}. 

We first review the related results known prior to this paper. The following conjecture was stated in \cite{CTLO}. 
\begin{conj}[Conjecture 1, \cite{CTLO}]
	For any side-information digraph $\mathcal{D}^k$, $k \in \{13,14,\cdots,25\}$, having any type of interaction (i.e., either fully-participated or partially-participated) between its sub-digraphs $\mathcal{D}_{i}^{k,\mathcal{P}}$, $i \in \{1,2,3\}$, for any $\mathcal{P}$, and $t$-bit messages for any finite $t$,
	\[    
	\beta_{t}(\mathcal{D}^{k},\mathcal{P})=\beta_{t}(\mathcal{D}_{1}^{k,\mathcal{P}})+\beta_{t}(\mathcal{D}_{2}^{k,\mathcal{P}})+\beta_{t}(\mathcal{D}_{3}^{k,\mathcal{P}})+\epsilon/t, 
	\]	   
	for some $\epsilon \in \{-2,-1,0\}$.
\end{conj}

The conjecture was stated considering that a minimum of  $\chi(\Gamma_t(\mathcal{D}_{1}^{k,\mathcal{P}}))\chi(\Gamma_t(\mathcal{D}_{2}^{k,\mathcal{P}}))\chi(\Gamma_t(\mathcal{D}_{3}^{k,\mathcal{P}}))$ ordered pairs of colors are required to color the confusion graph $\Gamma_t(\mathcal{D}^{k})$, $k \in \{1,2,\cdots,12\}$, according to the two-sender graph coloring. In this section, we show that there is a possibility to color the confusion graph  $\Gamma_t(\mathcal{D}^{k})$, $k \in \{16,18,20,21,23,25\}$, with comparitively less number of ordered pairs of colors. However, we do not provide an instance of the two-sender problem where our achievable broadcast rate is strictly less than that stated in the conjecture. The results are of importance as no non-trivial achievable broadcast rates with finite length messages are given for these cases in the literature.

The following achievable broadcast rate with $t$-bit messages for any finite $t$, for any two-sender problem belonging to Case II-E with fully-participated interactions was stated in Theorem 9 in \cite{CTLO} as an upper bound on $\beta_{t}(\mathcal{D}^k,\mathcal{P})$ with $k \in \{58,59,\cdots,64\}$. 

\begin{equation}
\begin{split} \beta_{t}(\mathcal{D}^k,\mathcal{P}) \leq & ~ ~ max(\beta_{t}(\mathcal{D}_{1}^{k,\mathcal{P}}),\beta_{t}(\mathcal{D}_{3}^{k,\mathcal{P}})) +\\ & max(\beta_{t}(\mathcal{D}_{2}^{k,\mathcal{P}}),\beta_{t}(\mathcal{D}_{3}^{k,\mathcal{P}})) .
\end{split}
\label{CTLOUBND}
\end{equation}

This result uses a code-construction based on any optimal codes (with $t$-bit messages) for the single-sender sub-problems. In this section, we provide a tighter upper bound for $\beta_{t}(\mathcal{D}^k,\mathcal{P})$ with $k \in \{58,59,\cdots,64\}$  by using another code-construction based on the same optimal codes for the single-sender sub-problems.

We first make the following observation which halves the number of sub-cases to be proved in Case I.

\begin{obs}
	Observe that the interaction digraphs $\mathcal{H}_k$, $k \in \{20,21,25\}$, are obtained from $\mathcal{H}_{k'}$, $k' \in \{16,18,23\}$, respectively, by interchanging the labels of vertices $1$ and $2$. If the corresponding TUICPs have the same set of sub-digraphs, i.e., $\mathcal{D}_1^{k,\mathcal{P}}$, $\mathcal{D}_2^{k,\mathcal{P}}$, and $\mathcal{D}_3^{k,\mathcal{P}}$ are same as $\mathcal{D}_1^{k',\mathcal{P}}$, $\mathcal{D}_2^{k',\mathcal{P}}$, and $\mathcal{D}_3^{k',\mathcal{P}}$ respectively, and  all the interactions are fully-participated interactions, then $\mathcal{D}^{k}$ can be obtained from $\mathcal{D}^{k'}$ by interchanging the labels of sub-digraphs $\mathcal{D}_1^{k',\mathcal{P}}$ and  $\mathcal{D}_2^{k',\mathcal{P}}$. Hence, an achievable broadcast rate for any TUICP with $\mathcal{H}_k$, $k \in \{20,21,25\}$, is obtained using that of a TUICP with $\mathcal{H}_{k'}$, $k' \in \{16,18,23\}$, respectively, by interchanging the labels $1$ and $2$ in the expression for the broadcast rate. 
	\label{obs1} 
\end{obs}

In the Theorems \ref{thmH1620}, \ref{thmH1821}, and \ref{thmH2325}, an achievable broadcast rate with finite length messages is obtained for any TUICP with fully-participated interactions between its sub-digraphs, based on a two-sender graph coloring of the confusion graph  $\Gamma_t(\mathcal{D}^{k})$ for $k \in \{16,18,23\}$. The results for any TUICP with side-information digraph $\mathcal{D}^{k}$, $k \in \{20,21,25\}$, are stated  without proof for completeness, based on Observation \ref{obs1}. 

\subsection{ An achievable broadcast rate for any TUICP with the side-information digraph being $\mathcal{D}^{16}$ or $\mathcal{D}^{20}$.} 

The following theorem provides an achievable broadcast rate with $t$-bit messages in terms of the corresponding optimal broadcast rates of the single-sender problems described by the side-information digraphs  $\mathcal{D}_{1}^{20,\mathcal{P}}$, $\mathcal{D}_{2}^{16,\mathcal{P}}$, $\mathcal{D}_{1*3}^{16,\mathcal{P}}$, and $\mathcal{D}_{2*3}^{20,\mathcal{P}}$. The notation  $\mathcal{D}_{u*v}^{k,\mathcal{P}}$ is explained in Definition \ref{defdisj1}.

\begin{thm}
	For any TUICP with the side-information digraph $\mathcal{D}^{k}$, $k \in \{16,20\}$, having fully-participated interactions between its sub-digraphs $\mathcal{D}_{i}^{k,\mathcal{P}}$, $i \in \{1,2,3\}$, for any $\mathcal{P}$, and $t$-bit messages for any finite $t$, the  following broadcast rates are achievable. 
	\begin{equation}
	(i) ~  p_{t}(\mathcal{D}^{16},\mathcal{P})=\beta_{t}(\mathcal{D}_{2}^{16,\mathcal{P}})+\beta_{t}(\mathcal{D}_{1*3}^{16,\mathcal{P}}). 
	\label{eqthmH16}
	\end{equation}
	\begin{equation}
	(ii) ~ p_{t}(\mathcal{D}^{20},\mathcal{P})=\beta_{t}(\mathcal{D}_{1}^{20,\mathcal{P}})+\beta_{t}(\mathcal{D}_{2*3}^{20,\mathcal{P}}). 
	\label{eqthmH20}
	\end{equation}
	\label{thmH1620}
\end{thm}
\begin{proof}
	See Appendix A for the proof of $(\ref{eqthmH16})$. The proof of $(\ref{eqthmH20})$ follows from the proof of $(\ref{eqthmH16})$ in conjunction with Observation \ref{obs1}.
\end{proof}

\begin{rem}
	Note that the proof of $(i)$ in Theorem \ref{thmH1620}, avails the following  symmetries of the confusion graph $\Gamma_{t}(\mathcal{D}^{16})$. All the $J$-blocks of $\Gamma_{t}(\mathcal{D}^{16})$ are isomorphic to each other. If there is an edge between any $j$th $J$-block and any $j'$th $J$-block, where $j,j' \in [2^{tm_2}]$, then there are edges between every vertex of the $j$th $J$-block and every vertex of the $j'$th $J$-block.
	\label{rem1} 
\end{rem}
The following lemmas, the first one stated as Lemma 10 in \cite{CTLO}, and the second one stated as Theorem 1 in \cite{NLUV} are required to prove our next result.

\begin{lem}[Lemma 10, \cite{CTLO}]
	For any real numbers $a$ and $b$, $\lceil a+b \rceil = \lceil a \rceil + \lceil b \rceil + \epsilon$, for some $\epsilon \in \{-1,0\}$.
	\label{addceil}
\end{lem}
\begin{lem}[Theorem 1, \cite{NLUV}] For any two undirected graphs $\mathcal{G}_1$ and $\mathcal{G}_2$,  
	$\chi(\mathcal{G}_1 * \mathcal{G}_2) \leq \chi(\mathcal{G}_1)\chi(\mathcal{G}_2)$.	
	\label{disjprod}
\end{lem}
The following Corollary \ref{corthmH1620} shows that there is a possibility to achieve a broadcast rate lesser than that stated by Conjecture 1 in \cite{CTLO}.
\begin{cor}
	For any TUICP with the side-information digraph $\mathcal{D}^{k}$, $k \in \{16,20\}$,  having fully-participated interactions between its sub-digraphs $\mathcal{D}_{i}^{k,\mathcal{P}}$, $i \in \{1,2,3\}$, for any $\mathcal{P}$, and $t$-bit messages for any finite $t$, we have,  
	\begin{equation}   
	p_{t}(\mathcal{D}^{k},\mathcal{P}) \leq \beta_{t}(\mathcal{D}_{1}^{k,\mathcal{P}})+\beta_{t}(\mathcal{D}_{2}^{k,\mathcal{P}})+\beta_{t}(\mathcal{D}_{3}^{k,\mathcal{P}})+\epsilon/t,
	\label{eqcorH1620}
	\end{equation}
	for some $\epsilon \in \{-1,0\}$, where $p_{t}(\mathcal{D}^{k},\mathcal{P})$ is the broadcast rate given in Theorem \ref{thmH1620}.
	\label{corthmH1620}
\end{cor}
\begin{proof}
	We first prove (\ref{eqcorH1620}) for $k = 16$. The proof of (\ref{eqcorH1620}) for $k=20$ follows from the proof for $k=16$ in conjunction with Observation \ref{obs1}.
	
	Consider the side-information digraph $\mathcal{D}_{1*3}^{16,\mathcal{P}}$ whose confusion graph is given by  $\Gamma_{t}(\mathcal{D}_{1}^{16,\mathcal{P}}) * \Gamma_{t}(\mathcal{D}_{3}^{16,\mathcal{P}})$. Using Lemma \ref{disjprod}  we have
	\begin{equation}
	\chi(\Gamma_{t}(\mathcal{D}_{1}^{16,\mathcal{P}}) * \Gamma_{t}(\mathcal{D}_{3}^{16,\mathcal{P}})) \leq \chi(\Gamma_{t}(\mathcal{D}_{1}^{16,\mathcal{P}})) \chi(\Gamma_{t}(\mathcal{D}_{3}^{16,\mathcal{P}})).
	\label{ineq1}
	\end{equation}
	Taking logarithm on both the sides of (\ref{ineq1}) and using Lemma \ref{addceil}, we have 
	\begin{equation}
	t\beta_{t}(\mathcal{D}_{1*3}^{16,\mathcal{P}}) \leq \lceil \log_2\chi(\Gamma_{t}(\mathcal{D}_{1}^{16,\mathcal{P}})) \rceil + \lceil \log_2\chi(\Gamma_{t}(\mathcal{D}_{3}^{16,\mathcal{P}})) \rceil + \epsilon,
	\label{ineq2} 
	\end{equation}
	for some $\epsilon \in \{-1,0\}$.
	We have used Lemma \ref{lemtwocolorrate} with the color set $\mathcal{J}_2=\Phi$ (which corresponds to a single-sender problem) to obtain   $t\beta_{t}(\mathcal{D}_{1*3}^{16,\mathcal{P}})=\lceil \log_2\chi(\Gamma_{t}(\mathcal{D}_{1*3}^{16,\mathcal{P}})) \rceil$ in (\ref{ineq2}). 
	Dividing both the sides of (\ref{ineq2}) by $t$ and using (\ref{eqthmH16}), we have the result of (\ref{eqcorH1620}). 
\end{proof}

\begin{rem}
	Note that there is a possibility of achieving a broadcast rate lesser than that stated by Conjecture 1 in \cite{CTLO}, when $p_{t}(\mathcal{D}^{k},\mathcal{P}) <  \beta_{t}(\mathcal{D}_{1}^{k,\mathcal{P}})+\beta_{t}(\mathcal{D}_{2}^{k,\mathcal{P}})+\beta_{t}(\mathcal{D}_{3}^{k,\mathcal{P}})+\epsilon/t$, where $\epsilon \in \{-1,0\}$. However, we do not provide any example, where the inequality holds strictly. The proof of the corollary suggests that the conjecture can be proved in negative if one can find two side-information digraphs $\mathcal{D}_{1}$ and $\mathcal{D}_{3}$ such that $\lceil \log_2\chi(\Gamma_{t}(\mathcal{D}_{1*3})) \rceil$	is strictly lesser than $\lceil \log_2\chi(\Gamma_{t}(\mathcal{D}_{1})) \rceil + \lceil \log_2\chi(\Gamma_{t}(\mathcal{D}_{3})) \rceil -2$.
\end{rem}

\subsection{An achievable broadcast rate for any TUICP with the side-information digraph being $\mathcal{D}^{18}$ or $\mathcal{D}^{21}$.}

The following lemma stated as Corollary 3.4.2 in \cite{SU} is required to derive our next result. Recall that the definition of lexicographic graph product denoted by $``\circ"$ is defined in Definition \ref{deflexi}.

\begin{lem}[Corollary 3.4.2, \cite{SU}]
	For any two undirected graphs $\mathcal{G}_1$ and $\mathcal{G}_2$,  
	$\chi(\mathcal{G}_1 \circ \mathcal{G}_2) \leq \chi(\mathcal{G}_1)\chi(\mathcal{G}_2)$.
	\label{lemlexi}
\end{lem}

The following theorem provides an achievable broadcast rate with $t$-bit messages in terms of the corresponding optimal broadcast rates of the single-sender problems described by the side-information digraphs  $\mathcal{D}_{1}^{21,\mathcal{P}}$, $\mathcal{D}_{2}^{18,\mathcal{P}}$, $\mathcal{D}_{1\circ3}^{18,\mathcal{P}}$, and $\mathcal{D}_{2\circ3}^{21,\mathcal{P}}$. Recall that the notation  $\mathcal{D}_{u \circ v}^{k,\mathcal{P}}$ is explained in Definition \ref{deflexi1}. The following theorem also avails the symmetries of the confusion graphs $\Gamma_{t}(\mathcal{D}^{k})$, $k \in \{18,21\}$,  to obtain the stated achievable broadcast rates (as given in Remark \ref{rem1}). 

\begin{thm}
	For any TUICP with the side-information digraph $\mathcal{D}^{k}$, $k \in \{18,21\}$, having fully-participated interactions between its sub-digraphs $\mathcal{D}_{i}^{k,\mathcal{P}}$, $i \in \{1,2,3\}$, for any $\mathcal{P}$, and $t$-bit messages for any finite $t$, the following broadcast rates are achievable, 
	\begin{equation}
	(i) ~ p_{t}(\mathcal{D}^{18},\mathcal{P})=\beta_{t}(\mathcal{D}_{2}^{18,\mathcal{P}})+\beta_{t}(\mathcal{D}_{1 \circ 3}^{18,\mathcal{P}}), 
	\label{eqthmH18}
	\end{equation}
	\begin{equation}
	(ii) ~ p_{t}(\mathcal{D}^{21},\mathcal{P})=\beta_{t}(\mathcal{D}_{1}^{21,\mathcal{P}})+\beta_{t}(\mathcal{D}_{2 \circ 3}^{21,\mathcal{P}}), 
	\label{eqthmH21}
	\end{equation}
	and for the same achievable broadcast rate $p_{t}(\mathcal{D}^{k},\mathcal{P})$, we have
	\begin{equation}
	(iii) ~
	p_{t}(\mathcal{D}^{k},\mathcal{P}) \leq \beta_{t}(\mathcal{D}_{1}^{k,\mathcal{P}})+\beta_{t}(\mathcal{D}_{2}^{k,\mathcal{P}})+\beta_{t}(\mathcal{D}_{3}^{k,\mathcal{P}})+\epsilon/t, 
	\label{eqcorH1821}
	\end{equation}	
	for some $\epsilon \in \{-1,0\}$.
	\label{thmH1821}
\end{thm}
\begin{proof}
	See Appendix B for the proof of (\ref{eqthmH18}). The proof of (\ref{eqthmH21}) follows from the proof of (\ref{eqthmH18}) in conjunction with Observation 1. The proof of (\ref{eqcorH1821}) follows from the proofs of (\ref{eqthmH18}) and (\ref{eqthmH21}), on the same lines as that of Corollary \ref{corthmH1620}, using Lemma \ref{lemlexi} instead of Lemma \ref{disjprod}.
\end{proof}

\subsection{An achievable broadcast rate for any TUICP with the side-information digraph being $\mathcal{D}^{23}$ or $\mathcal{D}^{25}$.}

The following theorem provides achievable broadcast rates for two sub-cases of Case I, availing the symmetries of the confusion graph as seen in Theorems \ref{thmH1620} and \ref{thmH1821}. 

\begin{thm}
	For any TUICP with the  side-information digraph $\mathcal{D}^{k}$, $k \in \{23,25\}$, having fully-participated interactions between its sub-digraphs $\mathcal{D}_{i}^{k,\mathcal{P}}$, $i \in \{1,2,3\}$, for any $\mathcal{P}$, and $t$-bit messages for any finite $t$, the following broadcast rates are achievable, 
	\begin{equation}
	(i) ~ p_{t}(\mathcal{D}^{23},\mathcal{P})=\beta_{t}(\mathcal{D}_{2}^{23,\mathcal{P}})+\beta_{t}(\mathcal{D}_{3 \circ 1}^{23,\mathcal{P}}),
	\label{eqthmH23}
	\end{equation}
	\begin{equation}
	(ii) ~ p_{t}(\mathcal{D}^{25},\mathcal{P})=\beta_{t}(\mathcal{D}_{1}^{25,\mathcal{P}})+\beta_{t}(\mathcal{D}_{3 \circ 2}^{25,\mathcal{P}}),
	\label{eqthmH25}
	\end{equation}
	and for the same achievable broadcast rate $p_{t}(\mathcal{D}^{k},\mathcal{P})$, we have
	\begin{equation}
	(iii) ~
	p_{t}(\mathcal{D}^{k},\mathcal{P}) \leq \beta_{t}(\mathcal{D}_{1}^{k,\mathcal{P}})+\beta_{t}(\mathcal{D}_{2}^{k,\mathcal{P}})+\beta_{t}(\mathcal{D}_{3}^{k,\mathcal{P}})+\epsilon/t, 
	\label{eqcorH2325}
	\end{equation}	
	for some $\epsilon \in \{-1,0\}$.
	\label{thmH2325}
\end{thm}
\begin{proof}
	See Appendix C for the proof of (\ref{eqthmH23}). The proof of (\ref{eqthmH25}) follows from the proof of (\ref{eqthmH23}) in conjunction with Observation 1. The proof of (\ref{eqcorH2325}) follows from the proofs of (\ref{eqthmH23}) and (\ref{eqthmH25}), on the same lines as that of Corollary \ref{corthmH1620}, using Lemma \ref{lemlexi} instead of Lemma \ref{disjprod}.
\end{proof}

\subsection{An achievable broadcast rate for any TUICP belonging to Case II-E with fully-participated interactions.}

The following theorem provides an achievable broadcast rate for any TUICP belonging to Case II-E, by providing a code construction which uses optimal codes of the sub-problems described by the three sub-digraphs of the side-information digraph. This provides a tighter upper bound compared to the one given in \cite{CTLO} and stated in (\ref{CTLOUBND}). 

\begin{thm}
	For any TUICP with the side-information digraph  $\mathcal{D}^{k}$, $k \in \{58,59,\cdots,64\}$, having fully-participated interactions between its sub-digraphs $\mathcal{D}_{i}^{k,\mathcal{P}}$, $i \in \{1,2,3\}$, for any $\mathcal{P}$, and $t$-bit messages for any finite $t$, the following broadcast rate is achievable.
	\begin{eqnarray}
	\begin{split}
	p_{t}(\mathcal{D}^{k},\mathcal{P}) & =  max\{\beta_{t}(\mathcal{D}_{1}^{k,\mathcal{P}})+\beta_{t}(\mathcal{D}_{2}^{k,\mathcal{P}}), \beta_{t}(\mathcal{D}_{1}^{k,\mathcal{P}}) \\ & +\beta_{t}(\mathcal{D}_{3}^{k,\mathcal{P}}),\beta_{t}(\mathcal{D}_{2}^{k,\mathcal{P}})+\beta_{t}(\mathcal{D}_{3}^{k,\mathcal{P}})\}.
	\end{split}
	\end{eqnarray}
	\label{thmH58_64a}
\end{thm}
\begin{proof}
	
	\par We provide a code-construction for $t$-bit messages for any finite $t$ and show that the constructed code satisfies all the demands of the receivers. For the case with  $\beta_{t}(\mathcal{D}_3^{k,\mathcal{P}}) \leq min\{\beta_{t}(\mathcal{D}_1^{k,\mathcal{P}}),\beta_{t}(\mathcal{D}_2^{k,\mathcal{P}})\}$, the broadcast rate $p_{t}(\mathcal{D}^{k},\mathcal{P}) =  \beta_{t}(\mathcal{D}_{1}^{k,\mathcal{P}})+\beta_{t}(\mathcal{D}_{2}^{k,\mathcal{P}})$, has been shown to be achievable in Theorem 9 of \cite{CTLO}. 
	
	Without loss of generality, we assume that $\beta_{t}(\mathcal{D}_2^{k,\mathcal{P}}) \leq min\{\beta_{t}(\mathcal{D}_1^{k,\mathcal{P}}),\beta_{t}(\mathcal{D}_3^{k,\mathcal{P}})\}$. The case with $\beta_{t}(\mathcal{D}_1^{k,\mathcal{P}}) \leq min\{\beta_{t}(\mathcal{D}_2^{k,\mathcal{P}}),\beta_{t}(\mathcal{D}_3^{k,\mathcal{P}})\}$ can be proved similarly. Let $\mathcal{C}_i$ be a code with the optimal broadcast rate with $t$-bit messages for any finite $t$ given by  $\beta_t(\mathcal{D}_i^{k,\mathcal{P}})$ for the single-sender unicast ICP described by $\mathcal{D}_i^{k,\mathcal{P}}$, $i \in \{1,2,3\}$. Our code for the original TUICP $\mathcal{I}(\mathcal{D}^k,\mathcal{P})$ is given as follows:
	\[  \mathcal{C}_1 \oplus \mathcal{C}_3[1:t\beta_{t}(\mathcal{D}_2^{k,\mathcal{P}})]  ~ ~ ~\mbox{sent by $\mathcal{S}_1$}, \]
	\[  \mathcal{C}_2 \oplus \mathcal{C}_3[1:t\beta_{t}(\mathcal{D}_2^{k,\mathcal{P}})]  ~ ~ ~ \mbox{sent by $\mathcal{S}_2$}, \]
	\[\mathcal{C}_3[1+t\beta_{t}(\mathcal{D}_2^{k,\mathcal{P}}):t\beta_{t}(\mathcal{D}_3^{k,\mathcal{P}})] ~ ~ ~ \mbox{sent by any one of $\mathcal{S}_1$ or $\mathcal{S}_2$}. \]    
	The overall length of the two-sender code is given by
	\[ t(\beta_{t}(\mathcal{D}_1^{k,\mathcal{P}})+\beta_{t}(\mathcal{D}_2^{k,\mathcal{P}})+(\beta_{t}(\mathcal{D}_3^{k,\mathcal{P}})-\beta_{t}(\mathcal{D}_2^{k,\mathcal{P}})))
	\]
	\[
	=t(\beta_{t}(\mathcal{D}_1^{k,\mathcal{P}})+\beta_{t}(\mathcal{D}_3^{k,\mathcal{P}})),
	\] with the broadcast rate $\beta_{t}(\mathcal{D}_1^{k,\mathcal{P}})+\beta_{t}(\mathcal{D}_3^{k,\mathcal{P}})$. 
	\par We provide the decoding procedure for receivers in the side-information digraphs $\mathcal{D}^k$ with $k \in \{58,59,\cdots,62\}$. The decoding procedure for those in the side-information digraphs $\mathcal{D}^k$ with $k \in \{63,64\}$ is similar. Receivers belonging to $\mathcal{D}_1^{k,\mathcal{P}}$ and $\mathcal{D}_2^{k,\mathcal{P}}$ recover their demanded messages using
	\[(\mathcal{C}_2 \oplus \mathcal{C}_3[1:t\beta_{t}(\mathcal{D}_2^{k,\mathcal{P}})])  \oplus (\mathcal{C}_1 \oplus \mathcal{C}_3[1:t\beta_{t}(\mathcal{D}_2^{k,\mathcal{P}})]) = \mathcal{C}_1 \oplus \mathcal{C}_2
	\]
	and their side-information $\mathcal{P}_2$ and $\mathcal{P}_1$ respectively. Receivers belonging to $\mathcal{D}_3^{k,\mathcal{P}}$ recover their demanded messages using $\mathcal{C}_3[t\beta_{t}(\mathcal{D}_2^{k,\mathcal{P}})+1:t\beta_{t}(\mathcal{D}_3^{k,\mathcal{P}})]$ and either  $\mathcal{C}_2 \oplus \mathcal{C}_3[1:t\beta_{t}(\mathcal{D}_2^{k,\mathcal{P}})]$ or $\mathcal{C}_1 \oplus \mathcal{C}_3[1:t\beta_{t}(\mathcal{D}_2^{k,\mathcal{P}})]$, and their side-information, depending on the presence of the interaction  $\mathcal{D}_3^{k,\mathcal{P}} \rightarrow \mathcal{D}_{2}^{k,\mathcal{P}}$ or $\mathcal{D}_3^{k,\mathcal{P}} \rightarrow \mathcal{D}_{1}^{k,\mathcal{P}}$ respectively. 
\end{proof}

\begin{rem}
	Note that the upper bound on  $\beta_{t}(\mathcal{D}^{k},\mathcal{P})$, $k \in \{58,59,\cdots,64\}$, stated in (\ref{CTLOUBND}) can also be written as follows.
	\begin{eqnarray}
	\begin{split}
	\beta_{t}(\mathcal{D}^{k},\mathcal{P}) & \leq  max\{\beta_{t}(\mathcal{D}_{1}^{k,\mathcal{P}})+\beta_{t}(\mathcal{D}_{2}^{k,\mathcal{P}}), \beta_{t}(\mathcal{D}_{1}^{k,\mathcal{P}})+ \\ &  \beta_{t}(\mathcal{D}_{3}^{k,\mathcal{P}}), \beta_{t}(\mathcal{D}_{2}^{k,\mathcal{P}})+\beta_{t}(\mathcal{D}_{3}^{k,\mathcal{P}}),2\beta_{t}(\mathcal{D}_{3}^{k,\mathcal{P}})\}. 
	\end{split}
	\end{eqnarray}
	Comparing this upper bound with the achievable broadcast rate given in Theorem \ref{thmH58_64a}, we see that the achievable broadcast rate given in Theorem \ref{thmH58_64a} is a tighter upper bound.
\end{rem}

\section{Optimal Broadcast Rates for Cases II-C, II-D, and II-E}

In this section, we provide the optimal broadcast rate for any TUICP with fully-participated interactions between the sub-digraphs of the side-information digraph $\mathcal{D}^k$, where $k \in \{34,35,\cdots,64\}$. Optimal broadcast rate for any TUICP with $\mathcal{D}^k$ such that $k \in \{1,2,\cdots,33\}$ were given in \cite{CTLO}. For $k \in \{34,35,\cdots,64\}$, results given in \cite{CTLO} depend on the relation between the optimal broadcast rates of the individual single-sender sub-problems described by the three sub-digraphs of the side-information digraph. The results given in this section along with those given in \cite{CTLO} provide a complete  characterisation of the optimal broadcast rate of any TUICP with fully-participated interactions.

We require the following lemma which is a part of Theorem 3 in \cite{tahmasbi2015critical} to derive our results. 
\begin{lem}[Theorem 3,  \cite{tahmasbi2015critical}]
	Consider any single-sender unicast index coding problem described by a side-information digraph. Removing edges not lying on any directed cycle does not change the optimal broadcast rate.
	\label{lemscc}
\end{lem}

We require the following lemmas to prove our results.
\begin{lem}
	For any $\mathcal{D}$, $\mathcal{P}$ and finite $t$, if a side-information digraph $\mathcal{D}'$ is obtained by adding more directed edges to $\mathcal{D}$, we have  $\beta_{t}(\mathcal{D},\mathcal{P}) \geq \beta_{t}(\mathcal{D}',\mathcal{P})$ and $\beta(\mathcal{D},\mathcal{P}) \geq \beta(\mathcal{D}',\mathcal{P})$.
	\label{addedges}
\end{lem}
\begin{proof}
	Consider an optimal code for the two-sender problem $\mathcal{I}(\mathcal{D},\mathcal{P})$ with $t$-bit messages with broadcast rate $\beta_{t}(\mathcal{D},\mathcal{P})$. This code can be used to solve the two-sender problem $\mathcal{I}(\mathcal{D}',\mathcal{P})$ with $t$-bit messages, as the receivers have additional side-information including the side-information present in the original problem $\mathcal{I}(\mathcal{D},\mathcal{P})$. Hence, $\beta_{t}(\mathcal{D},\mathcal{P}) \geq \beta_{t}(\mathcal{D}',\mathcal{P})$. Taking the limit as $t \rightarrow \infty$, in the definition of the optimal broadcast rate, we have $\beta(\mathcal{D},\mathcal{P}) \geq \beta(\mathcal{D}',\mathcal{P})$.
\end{proof}

\begin{lem}
	For any $\mathcal{D}$, $\mathcal{P}$ and finite $t$, we have  $\beta_{t}(\mathcal{D},\mathcal{P}) \geq \beta_{t}(\mathcal{D})$, and $\beta(\mathcal{D},\mathcal{P}) \geq \beta(\mathcal{D})$.
	\label{singlesenderlowbnd}
\end{lem}
\begin{proof}
	Consider a two-sender index code with broadcast rate given by $\beta_{t}(\mathcal{D},\mathcal{P})$. The same index code transmitted by a single-sender for the single-sender unicast ICP described by the  side-information digraph $\mathcal{D}$ satsifies the demands of all the receivers. Thus, we have the first lower bound. Taking the limit as $t \rightarrow \infty$, in the definition of the optimal broadcast rate, we have  $\beta(\mathcal{D},\mathcal{P}) \geq \beta(\mathcal{D})$. 
\end{proof}

\subsection{CASES II-C and II-D}
In this subsection, we provide the optimal broadcast rate for any TUICP with side-information digraph $\mathcal{D}^k$, where $k \in \{34,35,\cdots,57\}$.

\begin{thm}[CASE II-C]
	For any TUICP with the side-information digraph $\mathcal{D}^{k}$, $k \in \{34,35,\cdots,45\}$, having fully-participated interactions between its sub-digraphs $\mathcal{D}_{i}^{k,\mathcal{P}}$, $i \in \{1,2,3\}$, and for any $\mathcal{P}$, we have 
	\begin{equation} \beta(\mathcal{D}^{k},\mathcal{P})=max\{\beta(\mathcal{D}_{1}^{k,\mathcal{P}}),\beta(\mathcal{D}_{3}^{k,\mathcal{P}})\}+\beta(\mathcal{D}_{2}^{k,\mathcal{P}}).
	\end{equation}
	\label{thmH34_45}
\end{thm}
\begin{proof}
The result is proved in \cite{CTLO}, for the case when $\beta(\mathcal{D}_{1}^{k,\mathcal{P}}) \geq \beta(\mathcal{D}_{3}^{k,\mathcal{P}})$. Hence, we prove the result for the case with $\beta(\mathcal{D}_{1}^{k,\mathcal{P}}) < \beta(\mathcal{D}_{3}^{k,\mathcal{P}})$ by first providing a lower bound and then providing a matching upper bound.
 	
Removing the vertices belonging to $\mathcal{D}_1^{k,\mathcal{P}}$ from $\mathcal{D}^k$, we obtain a digraph $\mathcal{D}_{23}^{k,\mathcal{P}}$ which defines a TUICP. This can be considered as a single-sender unicast ICP as both $\mathcal{P}_2$ and $\mathcal{P}_{3}$ are with $S_2$. Hence, we have
\begin{equation}
	\beta(\mathcal{D}^k,\mathcal{P}) \geq \beta(\mathcal{D}_{23}^{k,\mathcal{P}}).
	\label{eqlem15}
\end{equation}
As there are only unidirectional edges from $\mathcal{V}(\mathcal{D}_2^{k,\mathcal{P}})$ to $\mathcal{V}(\mathcal{D}_{3}^{k,\mathcal{P}})$ or vice-versa (depending on the particular value of $k$), using Lemma \ref{lemscc}, we have 
\begin{equation}
	\beta(\mathcal{D}_{23}^{k,\mathcal{P}}) = \beta(\mathcal{D}_2^{k,\mathcal{P}}) +  \beta(\mathcal{D}_{3}^{k,\mathcal{P}}).
	\label{eqlem151}
\end{equation}
From (\ref{eqlem15}) and (\ref{eqlem151}), we have $\beta(\mathcal{D}^k,\mathcal{P}) \geq \beta(\mathcal{D}_2^{k,\mathcal{P}}) +  \beta(\mathcal{D}_{3}^{k,\mathcal{P}})$.

From the result of Theorem $8$ in \cite{CTLO}, we have,
\begin{equation}
\beta_t(\mathcal{D}^{k},\mathcal{P}) \leq \beta_t(\mathcal{D}_2^{k,\mathcal{P}}) +  \beta_t(\mathcal{D}_{3}^{k,\mathcal{P}}).
\label{eqlem1511}
\end{equation}
Dividing both the sides by $t$, and taking the  limit as $t \rightarrow \infty$ in (\ref{eqlem1511}), we have $  \beta(\mathcal{D}^{k},\mathcal{P}) \leq \beta(\mathcal{D}_2^{k,\mathcal{P}}) +  \beta(\mathcal{D}_{3}^{k,\mathcal{P}})$, which is a matching upper bound.  
\end{proof}

We make the following observation to obtain the optimal broadcast rate for Case II-D.

\begin{obs}
	Observe that the interaction digraphs $\mathcal{H}_k$, $k \in \{34,35,\cdots,45\}$ are obtained from $\mathcal{H}_k'$, $k' \in \{46,47,\cdots,57\}$, by interchanging the labels of vertices $1$ and $2$ respectively. Hence, the optimal broadcast rate for any TUICP with $\mathcal{H}_k$, $k \in \{34,35,\cdots,45\}$, is obtained using that of a TUICP with $\mathcal{H}_k'$, $k' \in \{46,47,\cdots,57\}$, respectively, by interchanging the labels $1$ and $2$ in the expression for the optimal broadcast rate. Note that the corresponding TUICPs must have the same set of sub-digraphs, i.e.,  $\mathcal{D}_1^{k,\mathcal{P}}$, $\mathcal{D}_2^{k,\mathcal{P}}$, and $\mathcal{D}_3^{k,\mathcal{P}}$ must be same as $\mathcal{D}_1^{k',\mathcal{P}}$, $\mathcal{D}_2^{k',\mathcal{P}}$, and $\mathcal{D}_3^{k',\mathcal{P}}$ respectively.
	\label{obs2} 
\end{obs}
Hence, we state the following theorem which follows from Theorem \ref{thmH34_45} in conjuncion with Observation \ref{obs2}.

\begin{thm}[CASE II-D]
	For any TUICP with the side-information digraph $\mathcal{D}^{k}$, $k \in \{46,47,\cdots,57\}$, having fully-participated interactions between its sub-digraphs $\mathcal{D}_{i}^{k,\mathcal{P}}$, $i \in \{1,2,3\}$, and for any $\mathcal{P}$, we have
	\begin{equation} \beta(\mathcal{D}^{k},\mathcal{P})=max\{\beta(\mathcal{D}_{2}^{k,\mathcal{P}}),\beta(\mathcal{D}_{3}^{k,\mathcal{P}})\}+\beta(\mathcal{D}_{1}^{k,\mathcal{P}}).
	\end{equation}
	\label{thmH46_57}
\end{thm}

\subsection{CASE II-E}
In this subsection, we will present our results for Case II-E. The proof uses the results of Case II-C and Case II-D derived in the previous subsection.

\begin{thm}[CASE II-E]
	For any TUICP with the side-information digraph  $\mathcal{D}^{k}$, $k \in \{58,59,\cdots,64\}$, having fully-participated interactions between its sub-digraphs $\mathcal{D}_{i}^{k,\mathcal{P}}$, $i \in \{1,2,3\}$, and for any $\mathcal{P}$, we have 
	\begin{eqnarray}
	\begin{split}
	& \beta(\mathcal{D}^{k},\mathcal{P})=  max\{\beta(\mathcal{D}_{1}^{k,\mathcal{P}})+\beta(\mathcal{D}_{2}^{k,\mathcal{P}}),\\ & \beta(\mathcal{D}_{1}^{k,\mathcal{P}})+\beta(\mathcal{D}_{3}^{k,\mathcal{P}}),\beta(\mathcal{D}_{2}^{k,\mathcal{P}})+\beta(\mathcal{D}_{3}^{k,\mathcal{P}})\}.
	\end{split}
	\end{eqnarray}
	\label{thmH58_64}
\end{thm}
\begin{proof}
	We first provide a lower bound using the results of Cases II-C and II-D. Then, we provide a matching upper bound using the result of Theorem \ref{thmH58_64a}. 
	
	Given any side-information digraph $\mathcal{D}^k$ with $k \in \{58,59,\cdots,64\}$, with fully-participated interactions  between its sub-digraphs $\mathcal{D}_{1}^{k,\mathcal{P}}$, $\mathcal{D}_{2}^{k,\mathcal{P}}$, and $\mathcal{D}_{3}^{k,\mathcal{P}}$, we can get $(i)$ one of the side-information digraphs $\mathcal{D}^{k'}, k' \in \{44,45\}$ and $(ii)$ one of the side-information digraphs $\mathcal{D}^{k''}, k'' \in \{56,57\}$ with the same sub-digraphs $\mathcal{D}_{1}^{k,\mathcal{P}}$, $\mathcal{D}_{2}^{k,\mathcal{P}}$, and $\mathcal{D}_{3}^{k,\mathcal{P}}$ having fully-participated interactions, by adding appropriate edges between the sub-digraphs of $\mathcal{D}^k$. From Lemma \ref{addedges}, we have,
	\begin{equation}
	\beta(\mathcal{D}^k,\mathcal{P}) \geq  \beta(\mathcal{D}^{k'},\mathcal{P}),
	\label{case2e_lbnd1}
	\end{equation}
	\begin{equation}
	\beta(\mathcal{D}^k,\mathcal{P}) \geq  \beta(\mathcal{D}^{k''},\mathcal{P})
	\label{case2e_lbnd2}.
	\end{equation}
	\par Combining the results of Theorem \ref{thmH34_45} and Theorem \ref{thmH46_57} using (\ref{case2e_lbnd1}) and (\ref{case2e_lbnd2}), we get,
	\begin{equation}
	\begin{split}
	& \beta(\mathcal{D}^k,\mathcal{P}) \geq max\{\beta(\mathcal{D}_{1}^{k,\mathcal{P}})+\beta(\mathcal{D}_3^{k,\mathcal{P}}), \\ & \beta(\mathcal{D}_2^{k,\mathcal{P}})+\beta(\mathcal{D}_3^{k,\mathcal{P}}),\beta(\mathcal{D}_1^{k,\mathcal{P}})+\beta(\mathcal{D}_2^{k,\mathcal{P}})\}.
	\end{split}
	\end{equation}
	
	\par Using the result of Theorem \ref{thmH58_64a}, we have 
	\begin{eqnarray}
	\begin{split}
	\beta_{t}(\mathcal{D}^{k},\mathcal{P}) & \leq  max\{\beta_{t}(\mathcal{D}_{1}^{k,\mathcal{P}})+\beta_{t}(\mathcal{D}_{2}^{k,\mathcal{P}}), \beta_{t}(\mathcal{D}_{1}^{k,\mathcal{P}})+ \\ &  \beta_{t}(\mathcal{D}_{3}^{k,\mathcal{P}}), \beta_{t}(\mathcal{D}_{2}^{k,\mathcal{P}})+\beta_{t}(\mathcal{D}_{3}^{k,\mathcal{P}})\}. 
	\end{split}
	\end{eqnarray}	
	Taking the limit as $t\to\infty$ in the definition of $\beta(\mathcal{D}^k,\mathcal{P})$, we obtain the matching upper bound as follows, and hence the theorem is proved.
	\begin{eqnarray}
	\begin{split}
	\beta(\mathcal{D}^{k},\mathcal{P}) & \leq  max\{\beta(\mathcal{D}_{1}^{k,\mathcal{P}})+\beta(\mathcal{D}_{2}^{k,\mathcal{P}}), \beta(\mathcal{D}_{1}^{k,\mathcal{P}})+ \\ &  \beta(\mathcal{D}_{3}^{k,\mathcal{P}}), \beta(\mathcal{D}_{2}^{k,\mathcal{P}})+\beta(\mathcal{D}_{3}^{k,\mathcal{P}})\}. 
	\end{split}
	\end{eqnarray}
\end{proof}

\begin{rem}
	We note that the results of all the theorems in this section are given in terms of those of the sub-problems which are single-sender unicast ICPs. However, the optimal broadcast rates of single-sender unicast ICPs are known only for some special cases \cite{SUOH}, \cite{sasi2018structure},\cite{MBSR}. Hence, the complexity of solving the two-sender problem is reduced to that of solving single-sender problems. 
\end{rem}

\begin{rem}
	For Case II-E, \cite{CTLO} provided  upper bound for the optimal broadcast rate with $t$-bit messages when $\beta_{t}(\mathcal{D}_3^{k,\mathcal{P}}) > min\{\beta_{t}(\mathcal{D}_1^{k,\mathcal{P}}),\beta_{t}(\mathcal{D}_2^{k,\mathcal{P}})\}$ and optimal broadcast rate when $\beta(\mathcal{D}_3^{k,\mathcal{P}}) > min\{\beta(\mathcal{D}_1^{k,\mathcal{P}}),\beta(\mathcal{D}_2^{k,\mathcal{P}})\}$. However, we have shown that the given upper bounds in \cite{CTLO} are loose, and Theorem \ref{thmH58_64} provides the optimal broadcast rates for Case II-E. 
\end{rem}

\section{Conclusion and Future Work}   
\par This paper establishes the optimal broadcast rates for all the cases of the TUICP with fully-participated interactions, for which only upper bounds were known. The results are given in terms of those of the  three single-sender sub-problems. Achievable broadcast rate with $t$-bit messages for any finite $t$ is given for some cases of the TUICP with fully-participated interactions, using two-sender graph coloring of the confusion graph. No results were known for these cases.
\par We conjecture that the achievable broadcast rates with $t$-bit messages for any finite $t$, for the six sub-cases of Case I given in this paper are optimal.  
\par Finding non-trivial achievable broadcast rate with $t$-bit messages for any finite $t$, for the remaining sub-cases of Case I is an interesting problem. Optimal broadcast rates with partially-participated interactions is also open. Further, extension of the results to general number of senders is open. 

\section*{Acknowledgment}
This work was supported partly by the Science and Engineering Research Board (SERB) of Department of Science and Technology (DST), Government of India, through J.C. Bose National Fellowship to B. S. Rajan and VAJRA Fellowship to V. Aggarwal.

\appendices 

\section{Proof of Theorem 1}

\begin{proof} 
	\par To prove this theorem, we construct the confusion graph $\Gamma_{t}(\mathcal{D}^{16})$ and identify the edges due to confusions at the vertices (receivers) belonging to each of the sub-digraphs $\mathcal{D}_1^{16,\mathcal{P}}$, $\mathcal{D}_2^{16,\mathcal{P}}$, and $\mathcal{D}_3^{16,\mathcal{P}}$. Then, we exploit the symmetries of the confusion graph to color it according to the two-sender graph coloring. The number of ordered pairs of colors required to color the confusion graph is used to calculate an achievable broadcast rate with $t$-bit messages.
	\par To avail the symmetries of the confusion graph, we view  $\Gamma_{t}(\mathcal{D}^{16})$ as the union of all the $J$-blocks connected by inter-$J$-block edges. Next, we list all the edges of $\Gamma_{t}(\mathcal{D}^{16})$.
	\par \underline{Edges due to confusions at the vertices  in $\mathcal{V}(\mathcal{D}_1^{16,\mathcal{P}})$}:
	If ${\bf{b}}_{\mathcal{P}_{1}}^{i}$ and ${\bf{b}}_{\mathcal{P}_{1}}^{i'}$ are confusable at some vertex in  $\mathcal{V}(\mathcal{D}_1^{16,\mathcal{P}})$, $i,i' \in [2^{tm_1}]$, then the corresponding edges in $\Gamma_{t}(\mathcal{D}^{16})$ due to the confusion at the same vertex in $\mathcal{V}(\mathcal{D}^{16})$ are of the form ($({\bf{b}}_{\mathcal{P}_{1}}^{i},{\bf{b}}_{\mathcal{P}_{2}}^{j},{\bf{b}}_{\mathcal{P}_{3}}^{k}),({\bf{b}}_{\mathcal{P}_{1}}^{i'},{\bf{b}}_{\mathcal{P}_{2}}^{j},{\bf{b}}_{\mathcal{P}_{3}}^{k'})$), where  $j \in [2^{tm_2}]$, and $k,k' \in [2^{tm_3}]$, as the vertex has all the messages represented by $\mathcal{V}(\mathcal{D}_2^{16,\mathcal{P}})$ as its side-information in  $\mathcal{V}(\mathcal{D}^{16})$, and has no side-information belonging to $\mathcal{V}(\mathcal{D}_3^{16,\mathcal{P}})$ in $\mathcal{V}(\mathcal{D}^{16})$. Hence, confusion at any vertex in $\mathcal{V}(\mathcal{D}_1^{16,\mathcal{P}})$ does not contribute to inter-$J$-block edges.
	\par \underline{Edges due to confusions at the vertices in $\mathcal{V}(\mathcal{D}_2^{16,\mathcal{P}})$}:
	If ${\bf{b}}_{\mathcal{P}_{2}}^{j}$ and ${\bf{b}}_{\mathcal{P}_{2}}^{j'}$ are confusable at some vertex in  $\mathcal{V}(\mathcal{D}_2^{16,\mathcal{P}})$, $j,j' \in [2^{tm_2}]$, then the corresponding edges in $\Gamma_{t}(\mathcal{D}^{16})$ due to the confusion at the same vertex in $\mathcal{V}(\mathcal{D}^{16})$ are of the form ($({\bf{b}}_{\mathcal{P}_{1}}^{i},{\bf{b}}_{\mathcal{P}_{2}}^{j},{\bf{b}}_{\mathcal{P}_{3}}^{k}),({\bf{b}}_{\mathcal{P}_{1}}^{i'},{\bf{b}}_{\mathcal{P}_{2}}^{j'},{\bf{b}}_{\mathcal{P}_{3}}^{k'})$), where  $i,i' \in [2^{tm_1}]$, and $k,k' \in [2^{tm_3}]$, as the vertex has no side-information in  $\mathcal{V}(\mathcal{D}^{16})$ belonging to $\mathcal{V}(\mathcal{D}_{1}^{16,\mathcal{P}})$ and $\mathcal{V}(\mathcal{D}_{3}^{16,\mathcal{P}})$. Hence, confusion at any vertex in $\mathcal{V}(\mathcal{D}_2^{16,\mathcal{P}})$ results in inter-$J$-block edges.
	\par \underline{Edges due to confusions at the vertices in $\mathcal{V}(\mathcal{D}_3^{16,\mathcal{P}})$}: Confusion at any vertex in $\mathcal{V}(\mathcal{D}_3^{16,\mathcal{P}})$ does not result in  inter-$J$-block edges, as each vertex has all the messages represented by $\mathcal{V}(\mathcal{D}_2^{16,\mathcal{P}})$ as its  side-information. The edges are of the form ($({\bf{b}}_{\mathcal{P}_{1}}^{i},{\bf{b}}_{\mathcal{P}_{2}}^{j},{\bf{b}}_{\mathcal{P}_{3}}^{k}),({\bf{b}}_{\mathcal{P}_{1}}^{i'},{\bf{b}}_{\mathcal{P}_{2}}^{j},{\bf{b}}_{\mathcal{P}_{3}}^{k'})$), where  ${\bf{b}}_{\mathcal{P}_{3}}^{k}$ and ${\bf{b}}_{\mathcal{P}_{3}}^{k'}$ are confusable at some receiver in  $\mathcal{V}(\mathcal{D}_3^{16,\mathcal{P}})$, as there is no side-information belonging to $\mathcal{V}(\mathcal{D}_1^{16,\mathcal{P}})$ in $\mathcal{V}(\mathcal{D}^{16})$.
	\par \underline{Coloring the confusion graph $\Gamma_{t}(\mathcal{D}^{16})$}: From Lemma \ref{lemblockiso}, we know that all the $J$-blocks are isomorphic to each other. From the enlisting of all the edges of the confusion graph, we know that the  inter-$J$-block edges between any two  $J$-blocks are only due to the confusions at the  receivers belonging to  $\mathcal{V}(\mathcal{D}_2^{16,\mathcal{P}})$. Confusion at any receiver in $\mathcal{V}(\mathcal{D}_2^{16,\mathcal{P}})$ does not result in any edge belonging to any $J$-block. Hence, in order to color the confusion graph according to the two-sender graph coloring, we find an optimal classical graph coloring of any $J$-block and associate the resulting colors with sender $\mathcal{S}_1$. This can be done, as the edges within any $J$-block are only due to the confusions at the vertices belonging to  $\mathcal{V}(\mathcal{D}_1^{16,\mathcal{P}} \cup \mathcal{D}_3^{16,\mathcal{P}})$, and $\mathcal{S}_1$ alone has all the messages in $\mathcal{P}_1 \cup \mathcal{P}_3$. As all the $J$-blocks are isomorphic to each other and all the inter-$J$-block edges between any two  $J$-blocks are only due to  the confusions at the receivers belonging to  $\mathcal{V}(\mathcal{D}_2^{16,\mathcal{P}})$, the same set of colors can be used by $\mathcal{S}_1$ to color every $J$-block identically. This resolves all the  confusions at all the receivers in $\mathcal{V}(\mathcal{D}_1^{16,\mathcal{P}} \cup \mathcal{D}_3^{16,\mathcal{P}})$. 
	
	From the listing of all the edges in $\Gamma_{t}(\mathcal{D}^{16})$, we observe that there is an edge given by  $(({\bf{b}}_{\mathcal{P}_{1}}^{i},{\bf{b}}_{\mathcal{P}_{2}}^{j},{\bf{b}}_{\mathcal{P}_{3}}^{k}),({\bf{b}}_{\mathcal{P}_{1}}^{i'},{\bf{b}}_{\mathcal{P}_{2}}^{j},{\bf{b}}_{\mathcal{P}_{3}}^{k'}))$, belonging to any $j$th $J$-block iff either the edge $({\bf{b}}_{\mathcal{P}_{1}}^{i},{\bf{b}}_{\mathcal{P}_{1}}^{i'}) \in \Gamma_t(\mathcal{V}(\mathcal{D}_1^{16,\mathcal{P}}))$ or the edge $({\bf{b}}_{\mathcal{P}_{3}}^{k},{\bf{b}}_{\mathcal{P}_{3}}^{k'}) \in \Gamma_t(\mathcal{V}(\mathcal{D}_3^{16,\mathcal{P}}))$. From the definition of the disjunctive graph product, we observe that each $J$-block is isomorphic to $\Gamma_{t}(\mathcal{D}_{1}^{16,\mathcal{P}})*\Gamma_{t}(\mathcal{D}_{3}^{16,\mathcal{P}})$. Note that the graph $\Gamma_{t}(\mathcal{D}_{1}^{16,\mathcal{P}})*\Gamma_{t}(\mathcal{D}_{3}^{16,\mathcal{P}})$ and any $J$-block have $2^{tm_1} \times 2^{tm_3}$ vertices.  Hence, $\mathcal{S}_1$ requires a minimum of $\chi(\Gamma_{t}(\mathcal{D}_{1*3}^{16,\mathcal{P}}))$ colors to color any $J$-block.
	
	The confusions associated with inter-$J$-block edges between any two $J$-blocks can be resolved by $\mathcal{S}_2$ alone, as all such confusions are associated with  vertices in $\mathcal{V}(\mathcal{D}_2^{16,\mathcal{P}})$ and only $\mathcal{S}_2$ has all the messages in $\mathcal{P}_2$. Observe that there are inter-$J$-block edges between any $j$th and any $j'$th $J$-blocks iff $({\bf{b}}_{\mathcal{P}_{2}}^{j}, {\bf{b}}_{\mathcal{P}_{2}}^{j'})$ is an edge in $\Gamma_{t}(\mathcal{D}_{2}^{16,\mathcal{P}})$. We know that a minimum of $\chi(\Gamma_t(\mathcal{D}_2^{16,\mathcal{P}}))$ colors are required to color $\Gamma_{t}(\mathcal{D}_{2}^{16,\mathcal{P}})$. By assigning the color given to ${\bf{b}}_{\mathcal{P}_{2}}^{j}$ in $\Gamma_{t}(\mathcal{D}_{2}^{16,\mathcal{P}})$ to the $j$th $J$-block (to all the vertices in the $j$th $J$-block) for all $j \in [2^{tm_2}]$, we observe that all the confusions associated with all the inter-$J$-block edges are resolved. Hence, a minimum of $\chi(\Gamma_t(\mathcal{D}_2^{16,\mathcal{P}}))$ colors are sufficient for $\mathcal{S}_2$ to color the confusion graph.	
	
	Hence, this is a valid two-sender graph coloring of $\Gamma_{t}(\mathcal{D}^{16})$ requiring a total of $\chi(\Gamma_{t}(\mathcal{D}_{1*3}^{16,\mathcal{P}})) \times \chi(\Gamma_{t}(\mathcal{D}_{2}^{16,\mathcal{P}}))$ ordered pairs of colors, where $\mathcal{S}_1$ requires $\chi(\Gamma_{t}(\mathcal{D}_{1*3}^{16,\mathcal{P}}))$ colors and $\mathcal{S}_2$ requires $\chi(\Gamma_{t}(\mathcal{D}_{2}^{16,\mathcal{P}}))$ colors.
	
	\par Thus, we have the total length of the two-sender index code given by the sum of the lengths of codewords transmitted by the two-senders as, 
	\begin{eqnarray}
	\begin{split}
	& t \times p_t(\mathcal{D}^{16},\mathcal{P}) \\ & =  
	\lceil \log_2(\chi(\Gamma_{t}(\mathcal{D}_{1*3}^{16,\mathcal{P}}))) \rceil + \lceil \log_2(\chi(\Gamma_{t}(\mathcal{D}_{2}^{16,\mathcal{P}}))) \rceil.
	\end{split}
	\end{eqnarray}
	
	Hence, we have the associated broadcast rate given by 
	\begin{equation}
	p_t(\mathcal{D}^{16},\mathcal{P}) = \beta_{t}(\mathcal{D}_{1*3}^{16,\mathcal{P}}) + \beta_{t}(\mathcal{D}_{2}^{16,\mathcal{P}}).
	\end{equation}
	
\end{proof}

\section{Proof of Theorem 2}

\begin{proof} 
	\par To prove this theorem, we follow the same approach used to prove Theorem \ref{thmH1620}. 
	\par To avail the symmetries of the confusion graph, we view  $\Gamma_{t}(\mathcal{D}^{18})$ as the union of all the $J$-blocks connected by inter-$J$-block edges. We list all the edges of $\Gamma_{t}(\mathcal{D}^{18})$ as follows. 
	\par \underline{Edges due to confusions at the vertices  in $\mathcal{V}(\mathcal{D}_{1}^{18,\mathcal{P}})$}:
	If ${\bf{b}}_{\mathcal{P}_{1}}^{i}$ and ${\bf{b}}_{\mathcal{P}_{1}}^{i'}$ are confusable at some vertex in  $\mathcal{V}(\mathcal{D}_{1}^{18,\mathcal{P}})$, where $i,i' \in [2^{tm_1}]$, then the corresponding edges in $\Gamma_{t}(\mathcal{D}^{18})$ due to the confusion at the same vertex in $\mathcal{V}(\mathcal{D}^{18})$ are of the form ($({\bf{b}}_{\mathcal{P}_{1}}^{i},{\bf{b}}_{\mathcal{P}_{2}}^{j},{\bf{b}}_{\mathcal{P}_{3}}^{k}),({\bf{b}}_{\mathcal{P}_{1}}^{i'},{\bf{b}}_{\mathcal{P}_{2}}^{j},{\bf{b}}_{\mathcal{P}_{3}}^{k'})$), where  $j \in [2^{tm_2}]$, and $k,k' \in [2^{tm_3}]$, as the vertex has all the messages represented by $\mathcal{V}(\mathcal{D}_{2}^{18,\mathcal{P}})$ as its  side-information in  $\mathcal{V}(\mathcal{D}^{18})$, and has no side-information in $\mathcal{V}(\mathcal{D}^{18})$ belonging to $\mathcal{V}(\mathcal{D}_{3}^{18,\mathcal{P}})$. Hence, confusion at any vertex in $\mathcal{V}(\mathcal{D}_{1}^{18,\mathcal{P}})$ does not contribute to inter-$J$-block edges.
	\par \underline{Edges due to confusions at the vertices in $\mathcal{V}(\mathcal{D}_{2}^{18,\mathcal{P}})$}:
	If ${\bf{b}}_{\mathcal{P}_{2}}^{j}$ and ${\bf{b}}_{\mathcal{P}_{2}}^{j'}$ are confusable at some vertex in  $\mathcal{V}(\mathcal{D}_{2}^{18,\mathcal{P}})$,  then the corresponding edges in $\Gamma_{t}(\mathcal{D}^{18})$ due to the confusion at the same vertex in $\mathcal{V}(\mathcal{D}^{18})$ are of the form 	 ($({\bf{b}}_{\mathcal{P}_{1}}^{i},{\bf{b}}_{\mathcal{P}_{2}}^{j},{\bf{b}}_{\mathcal{P}_{3}}^{k}),({\bf{b}}_{\mathcal{P}_{1}}^{i'},{\bf{b}}_{\mathcal{P}_{2}}^{j'},{\bf{b}}_{\mathcal{P}_{3}}^{k'})$), where  $i,i' \in [2^{tm_1}]$, and $k,k' \in [2^{tm_3}]$, as the vertex has no  side-information belonging to $\mathcal{V}(\mathcal{D}_{1}^{18,\mathcal{P}})$ and $\mathcal{V}(\mathcal{D}_{3}^{18,\mathcal{P}})$ in  $\mathcal{V}(\mathcal{D}^{18})$. Hence, confusion at any vertex in $\mathcal{V}(\mathcal{D}_{2}^{18,\mathcal{P}})$ results in inter-$J$-block edges.
	
	\par \underline{Edges due to confusions at the  vertices in $\mathcal{V}(\mathcal{D}_{3}^{18,\mathcal{P}})$}: 
	Confusion at any vertex in $\mathcal{V}(\mathcal{D}_{3}^{18,\mathcal{P}})$ does not result in  inter-$J$-block edges, as each vertex has all the messages represented by $\mathcal{V}(\mathcal{D}_{1}^{18,\mathcal{P}})$ and $\mathcal{V}(\mathcal{D}_{2}^{18,\mathcal{P}})$ as its side-information in $\mathcal{V}(\mathcal{D}^{18})$. The edges are of the form $(({\bf{b}}_{\mathcal{P}_{1}}^{i},{\bf{b}}_{\mathcal{P}_{2}}^{j},{\bf{b}}_{\mathcal{P}_{3}}^{k}),({\bf{b}}_{\mathcal{P}_{1}}^{i},{\bf{b}}_{\mathcal{P}_{2}}^{j},{\bf{b}}_{\mathcal{P}_{3}}^{k'}))$, where  ${\bf{b}}_{\mathcal{P}_{3}}^{k}$ and ${\bf{b}}_{\mathcal{P}_{3}}^{k'}$ are confusable at some receiver in  $\mathcal{V}(\mathcal{D}_{3}^{18,\mathcal{P}})$.
	\par \underline{Coloring the confusion graph $\Gamma_{t}(\mathcal{D}^{18})$}: 
	We follow the same approach as that given in the proof of Theorem \ref{thmH1620} to color the confusion graph $\Gamma_{t}(\mathcal{D}^{18})$, as it can be easily verified that the same reasoning given for the coloring approach in the proof of Theorem \ref{thmH1620} also holds in this case. We only mention the required changes.   
	
	From the listing of edges in $\Gamma_{t}(\mathcal{D}^{18})$, we observe that there is an edge given by  $(({\bf{b}}_{\mathcal{P}_{1}}^{i},{\bf{b}}_{\mathcal{P}_{2}}^{j},{\bf{b}}_{\mathcal{P}_{3}}^{k}),({\bf{b}}_{\mathcal{P}_{1}}^{i'},{\bf{b}}_{\mathcal{P}_{2}}^{j},{\bf{b}}_{\mathcal{P}_{3}}^{k'}))$, belonging to any $j$th $J$-block iff either the edge $({\bf{b}}_{\mathcal{P}_{1}}^{i},{\bf{b}}_{\mathcal{P}_{1}}^{i'}) \in \Gamma_t(\mathcal{V}(\mathcal{D}_1^{18,\mathcal{P}}))$, or the edge $({\bf{b}}_{\mathcal{P}_{3}}^{k},{\bf{b}}_{\mathcal{P}_{3}}^{k'}) \in \Gamma_t(\mathcal{V}(\mathcal{D}_3^{18,\mathcal{P}}))$ and ${\bf{b}}_{\mathcal{P}_{1}}^{i}={\bf{b}}_{\mathcal{P}_{1}}^{i'}$. From the definition of the lexicographic graph product, we observe that each $J$-block is isomorphic to $\Gamma_{t}(\mathcal{D}_{1}^{18,\mathcal{P}}) \circ \Gamma_{t}(\mathcal{D}_{3}^{18,\mathcal{P}})$. Hence, $\mathcal{S}_1$ requires a minimum of $\chi(\Gamma_{t}(\mathcal{D}_{1 \circ 3}^{18,\mathcal{P}}))$ colors to color any $J$-block.
	
	As in the proof of Theorem \ref{thmH1620}, a minimum of $\chi(\Gamma_t(\mathcal{D}_2^{18,\mathcal{P}}))$ colors are sufficient for $\mathcal{S}_2$ to color the confusion graph.	
	
	Hence, this is a valid two-sender graph coloring of $\Gamma_{t}(\mathcal{D}^{18})$ requiring a total of $\chi(\Gamma_{t}(\mathcal{D}_{1 \circ  3}^{18,\mathcal{P}})) \times \chi(\Gamma_{t}(\mathcal{D}_{2}^{18,\mathcal{P}}))$ ordered pairs of colors, where $\mathcal{S}_1$ requires $\chi(\Gamma_{t}(\mathcal{D}_{1 \circ 3}^{18,\mathcal{P}}))$ colors and $\mathcal{S}_2$ requires $\chi(\Gamma_{t}(\mathcal{D}_{2}^{18,\mathcal{P}}))$ colors.
	
	\par Thus, we have the total length of the two-sender index code given by the sum of the lengths of the codewords transmitted by the two-senders as, 
	\begin{eqnarray}
	\begin{split}
	& t \times p_t(\mathcal{D}^{18},\mathcal{P}) \\ & =  
	\lceil \log_2(\chi(\Gamma_{t}(\mathcal{D}_{1 \circ 3}^{18,\mathcal{P}}))) \rceil + \lceil \log_2(\chi(\Gamma_{t}(\mathcal{D}_{2}^{18,\mathcal{P}}))) \rceil.
	\end{split}
	\end{eqnarray}
	
	Hence, we have the associated broadcast rate given by 
	\begin{equation}
	p_t(\mathcal{D}^{18},\mathcal{P}) = \beta_{t}(\mathcal{D}_{1 \circ  3}^{18,\mathcal{P}}) + \beta_{t}(\mathcal{D}_{2}^{18,\mathcal{P}}).
	\end{equation}
	
\end{proof}

\section{Proof of Theorem 3}

\begin{proof} 
	\par To prove this theorem, we follow the same approach used to prove Theorem \ref{thmH1821}. 
	\par To avail the symmetries of the confusion graph, we view  $\Gamma_{t}(\mathcal{D}^{23})$ as the union of all the $J$-blocks connected by inter-$J$-block edges. We list all the edges of $\Gamma_{t}(\mathcal{D}^{23})$ as follows. 
	\par \underline{Edges due to confusions at the vertices  in $\mathcal{V}(\mathcal{D}_{1}^{23,\mathcal{P}})$}:
	If ${\bf{b}}_{\mathcal{P}_{1}}^{i}$ and ${\bf{b}}_{\mathcal{P}_{1}}^{i'}$ are confusable at some vertex in  $\mathcal{V}(\mathcal{D}_{1}^{23,\mathcal{P}})$, where $i,i' \in [2^{tm_1}]$, then the corresponding edges in $\Gamma_{t}(\mathcal{D}^{23})$ due to the confusion at the same vertex in $\mathcal{V}(\mathcal{D}^{23})$ are of the form ($({\bf{b}}_{\mathcal{P}_{1}}^{i},{\bf{b}}_{\mathcal{P}_{2}}^{j},{\bf{b}}_{\mathcal{P}_{3}}^{k}),({\bf{b}}_{\mathcal{P}_{1}}^{i'},{\bf{b}}_{\mathcal{P}_{2}}^{j},{\bf{b}}_{\mathcal{P}_{3}}^{k})$), where  $j \in [2^{tm_2}]$, and $k,k' \in [2^{tm_3}]$, as the vertex has all the messages represented by $\mathcal{V}(\mathcal{D}_{2}^{23,\mathcal{P}})$ and $\mathcal{V}(\mathcal{D}_{3}^{23,\mathcal{P}})$ as its  side-information in  $\mathcal{V}(\mathcal{D}^{23})$. Hence, confusion at any vertex in $\mathcal{V}(\mathcal{D}_{1}^{23,\mathcal{P}})$ does not contribute to inter-$J$-block edges.
	\par \underline{Edges due to confusions at the vertices in $\mathcal{V}(\mathcal{D}_{2}^{23,\mathcal{P}})$}:
	If ${\bf{b}}_{\mathcal{P}_{2}}^{j}$ and ${\bf{b}}_{\mathcal{P}_{2}}^{j'}$ are confusable at some vertex in  $\mathcal{V}(\mathcal{D}_{2}^{23,\mathcal{P}})$,  then the corresponding edges in $\Gamma_{t}(\mathcal{D}^{23})$ due to confusion at the same vertex in $\mathcal{V}(\mathcal{D}^{23})$ are of the form 	 ($({\bf{b}}_{\mathcal{P}_{1}}^{i},{\bf{b}}_{\mathcal{P}_{2}}^{j},{\bf{b}}_{\mathcal{P}_{3}}^{k}),({\bf{b}}_{\mathcal{P}_{1}}^{i'},{\bf{b}}_{\mathcal{P}_{2}}^{j'},{\bf{b}}_{\mathcal{P}_{3}}^{k'})$), where  $i,i' \in [2^{tm_1}]$, and $k,k' \in [2^{tm_3}]$, as the vertex has no  side-information belonging to $\mathcal{V}(\mathcal{D}_{1}^{23,\mathcal{P}})$ and $\mathcal{V}(\mathcal{D}_{3}^{23,\mathcal{P}})$ in  $\mathcal{V}(\mathcal{D}^{23})$. Hence, confusion at any vertex in $\mathcal{V}(\mathcal{D}_{2}^{23,\mathcal{P}})$ results in inter-$J$-block edges.
	
	\par \underline{Edges due to confusions at the  vertices in $\mathcal{V}(\mathcal{D}_{3}^{23,\mathcal{P}})$}: 
	Confusion at any vertex in $\mathcal{V}(\mathcal{D}_{3}^{23,\mathcal{P}})$ does not result in  any inter-$J$-block edges, as each vertex has all the messages represented by $\mathcal{V}(\mathcal{D}_{2}^{23,\mathcal{P}})$ as its side-information in $\mathcal{V}(\mathcal{D}^{23})$. The edges are of the form $(({\bf{b}}_{\mathcal{P}_{1}}^{i},{\bf{b}}_{\mathcal{P}_{2}}^{j},{\bf{b}}_{\mathcal{P}_{3}}^{k}),({\bf{b}}_{\mathcal{P}_{1}}^{i'},{\bf{b}}_{\mathcal{P}_{2}}^{j},{\bf{b}}_{\mathcal{P}_{3}}^{k'}))$, where  ${\bf{b}}_{\mathcal{P}_{3}}^{k}$ and ${\bf{b}}_{\mathcal{P}_{3}}^{k'}$ are confusable at some receiver in  $\mathcal{V}(\mathcal{D}_{3}^{23,\mathcal{P}})$.
	\par \underline{Coloring the confusion graph $\Gamma_{t}(\mathcal{D}^{23})$}: 
	We follow the same approach as that given in the proof of Theorem \ref{thmH1620} to color the confusion graph $\Gamma_{t}(\mathcal{D}^{23})$, as it can be easily verified that the same reasoning given for the coloring approach in the proof of Theorem \ref{thmH1620} also holds in this case. We only mention the required changes.   
	
	From the listing of edges in $\Gamma_{t}(\mathcal{D}^{23})$, we observe that there is an edge given by  $(({\bf{b}}_{\mathcal{P}_{1}}^{i},{\bf{b}}_{\mathcal{P}_{2}}^{j},{\bf{b}}_{\mathcal{P}_{3}}^{k}),({\bf{b}}_{\mathcal{P}_{1}}^{i'},{\bf{b}}_{\mathcal{P}_{2}}^{j},{\bf{b}}_{\mathcal{P}_{3}}^{k'}))$, belonging to any $j$th $J$-block iff either the edge $({\bf{b}}_{\mathcal{P}_{1}}^{i},{\bf{b}}_{\mathcal{P}_{1}}^{i'}) \in \Gamma_t(\mathcal{V}(\mathcal{D}_1^{23,\mathcal{P}}))$ and ${\bf{b}}_{\mathcal{P}_{3}}^{k}={\bf{b}}_{\mathcal{P}_{3}}^{k'}$, or the edge $({\bf{b}}_{\mathcal{P}_{3}}^{k},{\bf{b}}_{\mathcal{P}_{3}}^{k'}) \in \Gamma_t(\mathcal{V}(\mathcal{D}_3^{23,\mathcal{P}}))$. From the definition of the lexicographic graph product, we observe that each $J$-block is isomorphic to $\Gamma_{t}(\mathcal{D}_{3}^{23,\mathcal{P}}) \circ \Gamma_{t}(\mathcal{D}_{1}^{23,\mathcal{P}})$. Hence, $\mathcal{S}_1$ requires a minimum of $\chi(\Gamma_{t}(\mathcal{D}_{3 \circ 1}^{23,\mathcal{P}}))$ colors to color any $J$-block.
	
	As in the proof of Theorem \ref{thmH1620}, a minimum of $\chi(\Gamma_t(\mathcal{D}_2^{23,\mathcal{P}}))$ colors are sufficient for $\mathcal{S}_2$ to color the confusion graph.	
	
	Hence, this is a valid two-sender graph coloring of $\Gamma_{t}(\mathcal{D}^{23})$ requiring a total of $\chi(\Gamma_{t}(\mathcal{D}_{3 \circ  1}^{23,\mathcal{P}})) \times \chi(\Gamma_{t}(\mathcal{D}_{2}^{23,\mathcal{P}}))$ ordered pairs of colors, where $\mathcal{S}_1$ requires $\chi(\Gamma_{t}(\mathcal{D}_{3 \circ 1}^{23,\mathcal{P}}))$ colors and $\mathcal{S}_2$ requires $\chi(\Gamma_{t}(\mathcal{D}_{2}^{23,\mathcal{P}}))$ colors.
	
	\par Thus, we have the total length of the two-sender index code given by the sum of the lengths of the codewords transmitted by the two-senders as, 
	\begin{eqnarray}
	\begin{split}
	& t \times p_t(\mathcal{D}^{23},\mathcal{P}) \\ & =  
	\lceil \log_2(\chi(\Gamma_{t}(\mathcal{D}_{3 \circ 1}^{23,\mathcal{P}}))) \rceil + \lceil \log_2(\chi(\Gamma_{t}(\mathcal{D}_{2}^{23,\mathcal{P}}))) \rceil.
	\end{split}
	\end{eqnarray}
	
	Hence, we have the associated broadcast rate given by 
	\begin{equation}
	p_t(\mathcal{D}^{23},\mathcal{P}) = \beta_{t}(\mathcal{D}_{3 \circ  1}^{23,\mathcal{P}}) + \beta_{t}(\mathcal{D}_{2}^{23,\mathcal{P}}).
	\end{equation}
	
\end{proof}

\end{document}